\documentclass[11pt,a4paper]{amsart}
\usepackage[left=2.5cm,right=2.5cm,bottom=3cm,top=3cm,bindingoffset=0.5cm]{geometry}
\usepackage[utf8]{inputenc}
\usepackage[T1]{fontenc}
\usepackage{color}
\usepackage{amsmath,amssymb,amsthm,amsfonts}
\usepackage{mathtools}
\usepackage{fancyhdr}
\usepackage{csquotes}
\usepackage{bbm}
\usepackage{comment}
\usepackage{tikz}
\usepackage{vmargin}
\usepackage{setspace}
\usepackage{xcolor}
\usepackage{esint}

\definecolor{vertfonce}{rgb}{0.20, 0.46, 0.25}
\definecolor{rougefonce}{rgb}{0.64, 0.09, 0.20}
\usepackage[breaklinks=true,
            colorlinks=true,
            linkcolor=rougefonce,
            citecolor=vertfonce]{hyperref}
\usepackage{cleveref}
\usepackage{mathrsfs, enumerate}

\numberwithin{equation}{section}

\newcommand{\N}{\mathbb{N}}
\newcommand{\Z}{\mathbb{Z}}
\newcommand{\R}{\mathbb{R}}
\newcommand{\C}{\mathbb{C}}
\newcommand{\dd}{\mathrm{d}}

\newcommand{\G}{\mathbf G}
\newcommand{\T}{\mathbf T}
\newcommand{\one}{\mathbf{1}}

\newcommand{\X}{\mathbf X}
\newcommand{\M}{\mathbf M}
\newcommand{\Y}{\mathbf Y}

\newcommand{\Lh}{\mathscr L_{h}}

\newtheorem{theorem}{Theorem}
\newtheorem{lemma}[theorem]{Lemma}

\newtheorem{definition}[theorem]{Definition}

\title{Magnetic tunneling between disc-shaped obstacles}

\author[S. Fournais]{S\o ren Fournais}
\address[S. Fournais]{Department of Mathematics, University of Copenhagen, Universitetsparken 5, DK-2100 Copenhagen \O, Denmark}
\email{fournais@math.ku.dk}

\author[L. Morin]{Léo Morin}
\address[L. Morin]{Department of Mathematics, University of Copenhagen, Universitetsparken 5, DK-2100 Copenhagen \O, Denmark}
\email{lpdm@math.ku.dk }

\begin{document}

\begin{abstract}
In this paper we derive formulae for the semiclassical tunneling in the presence of a constant magnetic field in $2$ dimensions. The `wells' in the problem are identical discs with Neumann boundary conditions, so we study the magnetic Neumann Laplacian in the complement of a set of discs. We provide a reduction method to an interaction matrix, which works for a general configuration of obstacles. When there are two discs, we deduce an asymptotic formula for the spectral gap. When the discs are placed along a regular lattice, we derive an effective operator which gives rise to the famous Harper's equation.
Main challenges in this problem compared to recent results on magnetic tunneling are the fact that one-well ground states have non-trivial angular momentum which depends on the semiclassical parameter, and the existence of eigenvalue crossings.
\end{abstract}
\maketitle
\section{Main results}

\subsection{Introduction}

This article is devoted to the spectral analysis of the magnetic Laplacian
\[\Lh = \Big(-ih \partial_{x_1} + \frac{B x_2}{2} \Big)^2 + \Big(-ih \partial_{x_2} - \frac{B x_1}{2} \Big)^2,  \]
acting on domains of $\R^2$ of the form
\[ \Omega_V = \R^2 \setminus \Big( \bigcup_{\alpha \in V} D(\alpha,R) \Big), \]
where $R>0$ is a fixed radius, and $V$ is a discrete set of points. We assume that the obstacles $D(\alpha,R)$ are disjoint, and we denote by $\Lh^V$ the magnetic Neumann realization of $\Lh$ on $\Omega_V$.

When the magnetic field is very large -- which is equivalent to the semiclassical limit $h \to 0^+$ -- we know that there are bound states concentrating along the boundaries of the obstacles, which decay exponentially fast in the distance to the obstacles. In fact, the low-lying spectrum of $\Lh^V$ is a superposition of the spectra of each single-obstacle problem\footnote{When $V=\lbrace \alpha \rbrace$, we write $\Lh^\alpha$ instead of $\Lh^{\lbrace \alpha \rbrace}$ for simplicity.} $\Lh^\alpha$, up to exponentially small corrections. The aim of this article is to estimate these small terms, which are due to the interactions between the obstacles. This is an example of the tunneling effect in the presence of a magnetic field.
Tunneling with magnetic field has given rise to recent investigations in various contexts. It is especially interesting because the magnetic field affects tunneling in a different way than the standard tunneling between potential wells, which is well understood and explained in \cite{H88,HS} for instance. Magnetic tunneling is only fully understood in few cases, and only in dimension two. In several models, tunneling occurs along a curve on which an effective one dimensional operator can be derived (see \cite{BHR22} along the boundary of a domain, \cite{FHK22} along a discontinuity curve of the magnetic field, \cite{AA23} along a vanishing curve of the magnetic field). Other works \cite{FSW22,FMR,HK22,HKS23,M24} rely on radiality assumptions on the single well problem. This article is inspired by \cite{FMR,M24} where a sharp tunneling formula -- i.e. an estimate of the spectral gap -- between two radial magnetic wells is obtained. The non radial situation is still not understood.

Frank \cite{F} considered a similar problem of magnetic tunneling between obstacles. He deals with non-radial obstacles which are distributed along a periodic lattice, and obtains an upper bound on tunneling which is exponentially small but not sharp. One of the difficulties in this setting is the infinite number of interactions. His work was another motivation for us, and we indeed prove below sharp tunneling estimates in the case of disks, using the methods of \cite{FMR,M24} which are inspired by \cite{HS}. A related problem is the tight-binding reduction studied in \cite{SW22}. Some of our techniques to deal with infinitely many obstacles are inspired by \cite{F,SW22}. Note that, as in \cite{SW22}, our analysis allows for arbitrary configurations of obstacles, which do not have to satisfy any kind of periodicity.

The case of periodic obstacles is also similar to the setting of \cite{HS88}, where the tunneling is induced by a periodic potential in the presence of a homogeneous magnetic field. However in that work, the magnetic field has to be small enough, so that the tunneling effect is dominated by the potential. In \cite{HS88} it is also underlined how this continuous periodic model is effectively described by the Harper equation (a special case of the Almost-Mathieu operator). We recall this link in Section \ref{sec.harper} below.

Finally, we also mention the very recent work \cite{HK24}, where tunneling between two radial electric wells with Aharonov-Bohm potential is studied. In this setting, the one-well operator has the same structure as ours. The tunneling formula is similar to our Theorems \ref{thm.two.obstacles} and \ref{thm.two.obstacles.2}. However, the authors do not consider more than two wells.

\medskip
Estimating the interaction between the obstacles $\{ D(\alpha,R)\}_{\alpha \in V}$ requires a good understanding of the single-obstacle operator $\Lh^\alpha$. First of all, these operators are unitarily equivalent to each other. Indeed, if $\tau_\alpha^B$ denotes the magnetic translation acting on $\Psi \in L^2(\R^2)$ as\footnote{Here and in the rest of the article $\alpha \wedge x = \alpha_1 x_2 - \alpha_2 x_1$ is the standard wedge product on $\R^2$.}
\[ \tau_\alpha^B \Psi (x) = e^{\frac{iB}{2h} \alpha \wedge x} \Psi(x-\alpha) ,\]
then we have
\begin{equation}\label{eq.Lh.conjugate}
 (\tau_\alpha^B)^* \Lh^\alpha \tau_\alpha^B = \Lh^0.
 \end{equation}
We also mention the main properties of $\tau_\alpha^B$,
\[ \tau_\alpha^B \tau_\beta^B = e^{\frac{iB}{2h} \alpha \wedge \beta} \tau_{\alpha+\beta}^B = e^{\frac{iB}{h} \alpha \wedge \beta} \tau_{\beta}^B \tau_\alpha^B, \quad {\rm{and}} \quad (\tau_\alpha^B)^* = \tau_{-\alpha}^B.\]
From \eqref{eq.Lh.conjugate} we deduce that $\Lh^\alpha$ has the same spectrum as $\Lh^0$ and that the corresponding eigenfunctions are related by magnetic translations. The spectrum of $\Lh^0$ has non-trivial dependence on $B$ -- or equivalently on $h$. The first eigenvalues are described in \cite{FH}. They depend on the following oscillating function of $h$, 
\begin{equation}\label{def.eh}
e(h) = \inf_{m \in \Z} | \xi_h - m |, \qquad \xi_h = \frac{BR^2}{2 h} + \sqrt{\frac{\Theta_0 BR^2}{ h}} + \mathcal C_2,
\end{equation}
where $\Theta_0 \in (0,1)$ and $\mathcal C_2 \in \R$ are universal constants. To emphasize this dependence, we introduce the following notion.
\begin{definition}
Let $e_0 \in (-\frac 12, \frac 12]$. We say that $h \to 0$ along an $e_0$-sequence if $h$ follows a sequence $(h_n)_{n \in \N}$ with $\lim_{n \rightarrow 0} h_n = 0$, such that
\[ \xi_{h_n} - n \underset{n \to \infty}{\longrightarrow} e_0. \]
\end{definition}
In particular, if $h \to 0$ along an $e_0$-sequence then $e(h) \to |e_0|$. With this notation we have the following description of the first eigenvalues of $\Lh^0$.

\begin{theorem}[ \cite{FH} Spectrum of the single-obstacle operator]\label{thm.single}
There exist universal constants $\Theta_0 \in (0,1)$, $\mathcal C_1 >0$ and $\mathcal C_0$, $\mathcal C_2 \in \mathbb R$ such that the following holds. For all $B,R>0$ and $e_0 \in (-\frac 12, \frac 12]$ we have:
\begin{enumerate}
\item The first eigenvalue $\mu_1$ of $\Lh^0$ satisfies
\[ \mu_1 = \Theta_0 B h + \frac{ \mathcal C_1 \sqrt{B}}{R} h^{\frac 32} + \frac{3 \mathcal C_1 \sqrt{\Theta_0}}{R^2} h^2 \big( e_0^2 + \mathcal C_0 \big) + o(h^2), \]
when $h \to 0$ along an $e_0$-sequence.
\item The second eigenvalue $\mu_2$ of $\Lh^0$ satisfies
\[ \mu_2 = \Theta_0 B h + \frac{ \mathcal C_1 \sqrt{B}}{R} h^{\frac 32} + \frac{3 \mathcal C_1 \sqrt{\Theta_0}}{R^2} h^2 \big( (1-|e_0|)^2 + \mathcal C_0 \big) + o(h^2),\]
when $h \to 0$ along an $e_0$-sequence.
\item The third eigenvalue $\mu_3$ of $\Lh^0$ satisfies
\[ \mu_3 = \Theta_0 B h + \frac{ \mathcal C_1 \sqrt{B}}{R} h^{\frac 32} + \frac{3 \mathcal C_1 \sqrt{\Theta_0}}{R^2} h^2 \big( (1+|e_0|)^2 + \mathcal C_0 \big) + o(h^2),\]
when $h \to 0$ along an $e_0$-sequence.
\end{enumerate}
\end{theorem}

Note that $|e_0|$, $1-|e_0|$ and $1+|e_0|$ are respectively the distance of $\xi_h$ to the closest, second closest, and third closest integer.

\subsection{Tunneling between two obstacles}

Starting from Theorem \ref{thm.single}, we can study the interaction between two discs. We denote by $L>2R$ the distance between two centers $\alpha_\ell = \big( - \frac L2,0 \big)$ and $\alpha_r = \big( \frac L2,0 \big)$, and we consider $\Lh^V$ with $V= \lbrace \alpha_\ell, \alpha_r \rbrace$. If $|e_0| \neq \frac 12$, then the first two eigenvalues of $\Lh^0$ are separated by a spectral gap of order $h^2$ and thus the situation is similar to \cite{FMR,M24}: $\Lh^V$ has two eigenvalues of order $\mu_1$ which are exponentially close to each other.

\begin{theorem}[Two obstacles with $|e_0|\neq \frac 12$] \label{thm.two.obstacles}
Let $B,R,L>0$ and $e_0 \in (-\frac 12, \frac 12)$. Assume $V = \lbrace \alpha_\ell, \alpha_r \rbrace$ and $L>6R$. Define
\begin{equation}\label{eq:Sh}
S_h(L) = \frac{BL}{4}\sqrt{L^2 - 4R^2} - \big( BR^2 + 2\sqrt{h\Theta_0 BR^2} \big) \ln \Big( \frac{L + \sqrt{L^2 - 4R^2}}{2R}\Big) 
\end{equation}
and 
$$
c_0 = \frac{3 \mathcal C_1\sqrt{\Theta_0}}{R^2}.
$$
There exist $C=C(B,R,L) >0$ and $K = K(B,R,L)>0$ such that the following holds. If $h$ is small enough along an $e_0$-sequence, the spectrum of $\Lh^V$ in $(- \infty, \mu_1 + \frac{c_0(1-2|e_0|)}{2} h^2 )$ consists of exactly two eigenvalues $\lambda_1 < \lambda_2$ such that
\[ \lambda_2 - \lambda_1= 2C K^{-2e_0} h e^{-\frac{S_h(L)}{h}}(1+ o(1)), \]
as $h\to 0$ along an $e_0$-sequence.
\end{theorem}

{\bf Remark.}
The factor $6$ in the condition $L > 6R$ is a technical artifact of our proof. We do not expect this to be optimal. However, the magnetic tunneling has an oscillating component on top of an exponential decay, and both contribute to the final $S_h$. In the proof, one needs to control that an error term where the oscillations are eliminated remains smaller than the main term. This requires the obstacles to be sufficiently far from each other. The estimate in question is considered in Appendix~\ref{sec.errors}.

Similar conditions occur in the following Theorems~\ref{thm.two.obstacles.2} and \ref{thm.lattice} (with a constant $25$ instead of $6$). The reason remains the same.
\medskip

When $|e_0|=\frac 12$, the situation is more complicated because the first eigenvalue $\mu_1$ may become double. In fact, it is known that $\mu_1$ is double for some arbitrarily small values of $h$. However, in this case 
there is a gap of order $h^2$ between the first two eigenvalues $\{\mu_1,\mu_2\}$ and the third eigenvalue. Using this gap, we can prove that there are four eigenvalues which are exponentially close to each other.

\begin{theorem}[Two obstacles with $e_0 = \frac 12$]\label{thm.two.obstacles.2} Let $B,R,L>0$ with $L> 6R$ and assume $V = \lbrace \alpha_\ell, \alpha_r \rbrace$. Let $h>0$ be small enough and such that $\mu_1(h) = \mu_2(h)$. Then the spectrum of $\Lh^V$ in $(- \infty, \mu_1 + c_0 h^2 )$ consists of four eigenvalues $\lambda_1 \leq \lambda_2 \leq \lambda_3 \leq \lambda_4$ such that
\begin{align*}
\lambda_2 - \lambda_1 &= C ( K^{\frac 12} + K^{- \frac 32}) h e^{- \frac{S_h(L)}{h}}( 1+o(1)),\\
\lambda_3 - \lambda_2 &= o \Big( h e^{- \frac{S_h(L)}{h}} \Big),\\
\lambda_4 - \lambda_3 &= C ( K^{\frac 12} + K^{- \frac 32}) h e^{- \frac{S_h(L)}{h}} ( 1+o(1)),
\end{align*}
with the same constants $C>0$ and $K>0$ as in Theorem~\ref{thm.two.obstacles}.
\end{theorem}

\textbf{Remark.} Note that the assumption $\mu_1=\mu_2$ implies that $|e_0|=\frac 12$ by Theorem \ref{thm.single}. In this case, there is no reason for the eigenvalues of $\Lh^V$ not to be double, since this is the case for the single obstacle. It is remarkable that for the double well operator, there is still a gap between $\lambda_2$ and $\lambda_1$. It is a natural question whether $\lambda_2 = \lambda_3$ or not. It could be that some hidden symmetry of $\Lh^0$ is somehow inherited by $\Lh^V$.

\subsection{Periodic obstacles}
Another interesting situation, we can deal with, is the case of obstacles on a square lattice, as in \cite{F}. When $V= L \Z^2$, $\Lh^V$ has infinitely many eigenvalues of order $\mu_1$. In Theorem \ref{thm.lattice} below, we prove that this part of the spectrum is given by an effective operator $\X$ on $\ell^2(L \Z^2)$, with coefficients
\[
\X_{\alpha \beta} = 
\begin{cases}
e^{\frac{iB}{2h}\alpha \wedge \beta}, &{\rm{if}} \, |\alpha - \beta |=L,\\
0, &{\rm{otherwise}}.
\end{cases}
\]
In the case $e_0 = \frac 12$, the effective operator $\Y$ acts on $\ell^2(L\Z^2) \otimes \C^2$ with coefficients
\[ 
\Y_{\alpha \beta}^{\sigma \sigma'}=
\begin{cases}
- \sigma' e^{\frac{iB}{2h}\alpha \wedge \beta} e^{- i \frac{\sigma - \sigma'}{2} \arg(\beta-\alpha)} K^{\frac{\sigma + \sigma'}{2}} ,&{\rm{if}} \, |\alpha - \beta|=L, \\
0, &{\rm{otherwise,}}
\end{cases}
\]
for $\alpha$, $\beta \in L\Z^2$ and $\sigma$, $\sigma' \in \lbrace \pm \rbrace$.
More precisely, we obtain the following result.

\begin{theorem}[Obstacles on a square lattice]\label{thm.lattice}
Let $B$, $R$, $L>0$ and $e_0 \in (- \frac 12, \frac 12]$. Assume $V = L \Z^2$, and $L> 25 R$. For $t>0$ define the spectral subspace
\[ \mathscr E(t) = \one_{(-\infty, \mu_1 + th^2]} ( \Lh^V ). \]
There exist $C=C(B,R,L) >0$ and $K = K(B,R,L)>0$ such that the following holds.
Assume that $h \rightarrow 0$ along an $e_0$-sequence.
\begin{enumerate}
\item If $e_0 \neq \frac 12$, then for $h$ small enough, the restriction of $\Lh^V$ to $\mathscr{E}\big( c_0 \frac{1-2e_0}{2} \big)$ is unitarily equivalent to a bounded operator $\M$ on $\ell^2(L\Z^2)$ such that
\[ \M = \mu_1 \mathbf I + \lambda \X + o(h e^{-\frac{S_h(L)}{h}}), \]
where $|\lambda| = 2C K^{-2e_0} h e^{- \frac{S_h(L)}{h}}$.
\item If $e_0= \frac 12$, then we assume furthermore that $\mu_1(h) = \mu_2(h)$ along the $e_0$-sequence. Then for $h$ small enough, the restriction of $\Lh^V$ to $\mathscr{E}(c_0)$ is unitarily equivalent to a bounded operator $\M$ on $\ell^2(L \Z^2) \otimes \mathbf C^2$ such that
\[ \M = \mu_1 \mathbf I + \lambda \Y+ o(h e^{-\frac{S_h(L)}{h}}),\]
where $\lambda = (-1)^{\lfloor \xi_h \rfloor}  2 C K^{- \frac 12} h e^{- \frac{S_h(L)}{h}}$.
\end{enumerate}
\end{theorem}

In both cases the operator $\M$ is obtained by constructing an eigenbasis for $\Lh^V$ starting from the ground states of each single-obstacle operator $\Lh^\alpha$ (for $\alpha \in L \Z^2$).
 
 \subsection{Link with the Almost-Mathieu operator (or Harper equation).}\label{sec.harper}

The effective operator $\X$ from Theorem \ref{thm.lattice} is in fact unitarily equivalent to a direct integral of Almost-Mathieu operators. Indeed, we can first conjugate $\X$ by the unitary diagonal matrix ${\rm{diag}} \big( e^{i \frac{B}{2h} \alpha_1 \alpha_2} \big)_{\alpha \in V}$. Thus $\X$ is unitarily equivalent to $\X'$ with
\[ \X'_{\alpha \beta} = e^{i \frac{B}{2h}(\beta_1 + \alpha_1)(\beta_2 - \alpha_2)}, \qquad |\alpha-\beta|=L \]
This is a convolution in the variable $\alpha_2$, and it can be diagonalized using the following Fourier transform,
\[ (\mathscr{F} f )(\alpha_1,\alpha_2) = \frac{1}{\sqrt{2\pi}}\int_{0}^{2\pi} e^{-i \theta \alpha_2 /L } f_{\alpha_1}(\theta) \dd \theta, \quad f \in \ell^2(L \Z) \otimes L^2(0,2\pi), \]
with inverse
\[ (\mathscr{F}^{-1} g)(\alpha_1,\theta) = \frac 1 {\sqrt{2\pi}} \sum_{\alpha_2 \in L \Z} e^{i \theta \alpha_2 /L} g_{(\alpha_1,\alpha_2)}, \quad g \in \ell^2(L \Z^2). \]
Then we have for all $f \in \ell^2(L\Z) \otimes L^2(0,2\pi)$,
\[ (\mathscr{F}^{-1} \X' \mathscr{F} f)(\alpha_1, \theta) = f(\alpha_1 - L, \theta) + f(\alpha_1+L,\theta) + 2 \cos \big( \theta - \frac{BL}{h} \alpha_1 \big) f(\alpha_1,\theta). \]
Therefore, 
\[ \mathscr{F}^{-1} \X' \mathscr{F} = \int_{(0,2\pi)}^\oplus \mathscr{H}_{\theta, \frac{BL^2}{2 \pi h}} \dd \theta, \]
where $\mathscr{H}_{\theta,\phi}$ is the Almost-Mathieu (or Harper) operator on $\ell^2(\Z)$,
\[ (\mathscr{H}_{\theta,\phi} u)(n) = u(n-1) + u(n+1) +2 \cos(\theta + 2 \pi n \phi) u(n). \]
This operator has been object of many works (see for instance \cite{AJ06,HS89,J21,Puig}, and the recent surveys \cite{J-survey,butterfly}). It is well-known to have very singular dependence on $\phi$. In particular, the following properties can be transfered to $\X$.
\begin{itemize}
\item If $\phi$ is rational, the spectrum of $\mathscr{H}_{\theta,\phi}$ is purely absolutely continuous and independent of $\theta$.
\item If $\phi$ is irrational, the spectrum of $\mathscr{H}_{\theta,\phi}$ is purely singular continuous \cite{J21}.
\item If $\phi$ is irrational, the spectrum of $\mathscr{H}_{\theta,\phi}$ is a Cantor set with Lebesgue measure $0$ \cite{AJ06}.
\end{itemize}
The picture of the spectrum of $\mathscr{H}_{\theta,\phi}$ as function of $\phi$ leads to the famous Hofstadter's butterfly. Theorem \ref{thm.lattice} therefore states that, around $\mu_1$, the spectrum of $\Lh^V$ is given to main order by a Hofstadter's butterfly (as function of the flux $\phi=\frac{BL^2}{2\pi h}$).

\medskip
\textit{Remark.} 
Near the crossing points $\mu_1 = \mu_2$, Theorem \ref{thm.lattice} gives the effective operator $\Y$, which is more complicated than an Almost-Mathieu operator, since it involves an effective spin. Using the same transformations as for $\X$, we find that $\Y$ is conjugated to the operator
\[ \int_{(0,2\pi)}^\oplus \mathscr{H}'_{\theta, \frac{BL^2}{2\pi h}} \dd \theta, \]
where $\mathscr{H}_{\theta,\phi}' : \ell^2(\Z) \otimes \C^2 \to \ell^2(\Z) \otimes \C^2$ is defined by
\[ (\mathscr{H}_{\theta,\phi}' u)(n) = \mathcal T u(n+1) + \mathcal T^* u(n-1) + 2 \mathcal A_n u(n), \]
with
\[ \mathcal T = 
\begin{pmatrix}
-K & 1 \\ -1 & K^{-1}
\end{pmatrix},
\quad
\mathcal A_n = 
\begin{pmatrix}
- K \cos (2 \pi n \phi  - \theta) & \sin (2\pi n \phi - \theta) \\ \sin (2\pi n \phi - \theta) & K^{-1} \cos (2\pi n \phi - \theta)
\end{pmatrix}
.
 \]
 This operator is maybe linked to the unitary Almost-Mathieu operators studied in \cite{CFO} for instance. 

\subsection{Organization of the paper}
The analysis necessary for the proofs of all main theorems: Theorems~\ref{thm.two.obstacles}, \ref{thm.two.obstacles.2}, and \ref{thm.lattice}, is to a large extent overlapping: The first two theorems concern the case of $2$ obstacles and in the last the set of obstacles is an infinite lattice. Therefore, in Sections~\ref{sec:SingleObstacle}-\ref{sec.reduction2} we carry through this analysis in a unified way for a general set of obstacles. The first Section~\ref{sec:SingleObstacle} contains the relevant information about the one-obstacle problem, which is an extremely important input to the rest of the paper. A main observation is that using the rotational symmetry of the one-obstacle problem, one can understand ground states in terms of a special function---the Kummer function---in the radial variable. A similar insight was used in \cite{FSW22} and spurred much progress on special cases of tunneling with magnetic fields (see in particular \cite{HK22,HKS23,M24,FMR}), however a main novelty in the present work is the presence of non-zero, $h$-dependent, angular momentum in the ground state and also the accompanying crossing of eigenvalues.
Section~\ref{sec.reduction1} and Section~\ref{sec.reduction2} both contain the reduction to an interaction matrix, in the cases without eigenvalue crossing ($e_0 \neq \frac 12$) and near a crossing ($e_0 = \frac 12$) respectively. All of this analysis is for a general set $V$ of obstacles and therefore include both the case of only two obstacles and the case of an infinite lattice of them. Note also, that the periodicity of the obstacles in the case of infinitely many is not used in this part of the paper.

The Section~\ref{sec:tunneling_2obstacles} now gives the proofs of Theorems~\ref{thm.two.obstacles} and \ref{thm.two.obstacles.2}, i.e. the tunneling formulas for two obstacles in the cases away from crossings or exactly at the crossings respectively.
Section~\ref{sec:tunneling_lattice} contains the tunneling analysis for an infinite set of obstacles placed on a regular lattice.

The most computational parts of the analysis are delegated to the appendices, in order to lighten the rest of the presentation. However, these parts are also fundamental for the final estimates. In Appendix \ref{sec.phase}, we gather the main properties of the phases involved in several integrals, needed to use the Laplace method. In Appendix \ref{appendix:IC}, we calculate the interaction coefficients arising in the effective operators. Finally, we bound some error terms in Appendix \ref{sec.errors}.

\section*{Acknowledgements}

The authors thank Bernard Helffer and Rupert Frank for interesting discussions and encouraging the study of periodic obstacles. L.M. also thanks Bernard Helffer for sharing references including \cite{HS88}.

This work is funded by the European Union (via the ERC Advanced Grant MathBEC - 101095820). Views and opinions expressed are however those of the author only and do not necessarily reflect those of the European Union or the European Research Council. Neither the European Union nor the granting authority can be held responsible for them.

\section{The single obstacle operator \texorpdfstring{$\Lh^0$}{}}\label{sec:SingleObstacle}

\subsection{Fourier decomposition}

We gave a description of the first eigenvalues of $\Lh^0$ in Theorem \ref{thm.single}. Let us recall some further properties of this operator. Since we consider magnetic-Neumann boundary conditions, the domain of $\Lh^0$ is
\[ {\rm{Dom}} \big( \Lh^0 \big) = \lbrace \Psi \in H^1(\Omega^0)  : \quad \Lh \Psi \in L^2(\Omega^0), \quad \partial_r \Psi = 0 \quad {\rm{on }} \; \partial D(0,R) \rbrace, \]
where we recall that $\Omega^0 = \R^2 \setminus D(0,R)$. Here $\partial_r$ denotes the radial derivative. In radial coordinates $(r,\theta)$ the differential operator $\Lh$ is given by\footnote{With slight abuse of notation, we still denote by $\Lh$ the operator in radial coordinates.}
\begin{equation}
\Lh = - \frac{h^2}{r} \partial_r r \partial_r + \frac{1}{r^2}\Big(-ih \partial_\theta - \frac{B r^2}{2} \Big)^2.
\end{equation}
In particular, $\Lh^0$ commutes with $\partial_\theta$ and we can decompose according to Fourier modes. Therefore, it is unitarily equivalent to the direct sum of the family of operators
\begin{equation}
\mathscr{L}_{h,m} =- \frac{h^2}{r} \partial_r r \partial_r + \frac{1}{r^2}\Big(h m - \frac{B r^2}{2} \Big)^2, \qquad m \in \Z.
\end{equation}
For all radial functions $f$ we have $\mathscr{L}_{h,m} f = e^{im \theta}  \Lh^0 e^{-im\theta} f$. Thus, the spectrum of $\Lh^0$ is the union of the spectra of each $\mathscr{L}_{h,m}$. A more precise result follows from the discussion in \cite{FH} on the first eigenvalues.

\begin{lemma}\label{lem.m-m+}
Let $e_0 \in (-\frac 12, \frac 12]$. 
\begin{enumerate}
\item If $h$ is small enough along an $e_0$-sequence with $e_0 \neq \frac 12$, there is a unique integer $m_-$ closest to $\xi_h$, and it satisfies
\[ m_- = \xi_h - e_0 + o(1). \] 
With this $m$, the ground state energy of $\Lh^0$ is
\[ \mu_1 = \lambda_1( \mathscr{L}_{h,m_-}). \]
\item If $h$ is small enough along a $\frac 12$-sequence, the two closest integers $m_-$ and $m_+$ to $\xi_h$ satisfy
\[ \begin{cases}
m_- &= \xi_h - \frac 12 + o(1) \\
m_+ &= \xi_h + \frac 12 + o(1).
\end{cases} \]
Which one is the closest depends on $h$. The two smallest eigenvalues of $\Lh^0$ are given by
\[ \lbrace \mu_1 , \mu_2 \rbrace = \lbrace\mu_-, \mu_+ \rbrace, \]
where $\mu_{\pm} =  \lambda_1( \mathscr{L}_{h,m_\pm})$.
\end{enumerate}
\end{lemma}

In any case, we will always denote by $\Phi^{-}$ and $\Phi^+$ the eigenfunctions of $\Lh^0$ with respective momenta $m_-$ and $m_+ = m_- +1$. In other words,
\begin{equation}\label{def.phipm}
 e^{-im_{\pm} \theta} \Phi^\pm \quad {\text{is the ground state of}} \quad \mathscr{L}_{h,m_\pm}.
\end{equation}
We denote by $\mu_\pm$ the associated eigenvalue.

\subsection{Decay of the eigenfunctions}

A fundamental aspect of our analysis is that we know precisely how fast the eigenfunctions $\Phi^\pm$ decay away from $r=R$. The decay rate is given by the function
 \begin{equation}\label{eq:d_function}
 d(r) = \frac{B}{4}(r^2 - R^2) - \frac{BR^2}{2} \ln( r/R), 
 \end{equation}
 which we call \emph{Agmon distance} by analogy with the case of potential wells. Indeed, one has 
 \begin{equation}
 |\Phi^\pm(r) | \leq C_\eta e^{- \frac{(1-\eta) d(r)}{h}}, 
 \end{equation}
 for all $\eta>0$. We use these bounds all through the paper. They are consequences of the following lemma, recalling that $m_\pm \sim \frac{BR^2}{2h}$ as $h \to 0$.

\begin{lemma}\label{lem.decay}
The eigenfunctions $\Phi^\pm$ of $\Lh^0$, defined in \eqref{def.phipm}, are of the form
\begin{equation}\label{eq.Phipm2}
\Phi^\pm(r,\theta) = \Big( \frac{BR^2}{h} \Big)^{\frac{\Theta_0}{4}} \Big( \frac {re^{i\theta}} R \Big)^{m_{\pm}} e^{- \frac{B}{4h}(r^2 - R^2)} u_h(r), 
\end{equation}
where $u_h$ is a real-analytic function on ${\rm{Re}}(r^2) >0$, which is locally uniformly bounded with respect to $h$. Moreover,
\begin{itemize}
\item For all $r$ such that ${\rm{Re}}(r^2) >R^2$, $u_h(r)$ is asymptotically given by
\begin{equation}\label{eq.uhr}
u_h(r) =  \frac{K_0  \Gamma(\delta_0)}{R} \Big( \frac{2R^2}{r^2-R^2} \Big)^{\delta_0}  \Big( 1 + o \Big( \frac{h^{\frac 12}}{r^2-R^2} \Big) \Big), 
\end{equation}
as $h \to 0$, where $\delta_0 = \frac{1 - \Theta_0}{2}$ and $K_0>0$ is a universal constant.
\item For all $t>0$ there is a $v(t) >0$ such that
\begin{equation} \label{eq.uhR}
 u_h \Big(R+ t \sqrt{\frac hB} \Big) = \frac{K_0 v(t)}{R} \Big( \frac{h}{BR^2} \Big)^{\frac{\Theta_0 - 1}{4}} \big( 1 +o(1) \big)
 \end{equation}
 as $h \to 0$. In fact, $v$ is a solution to
 \[ -v''(t) + 2(t- \sqrt{\Theta_0}) v'(t) + v(t) = \Theta_0 v(t).\] 
 \item For all $t>0$,
 \begin{equation}\label{eq.uh'}
  \partial_r u_h\Big(R+t \sqrt{\frac h B} \Big) = \frac{K_0 v'(t)}{ R^2} \Big( \frac{h}{BR^2} \Big)^{\frac{\Theta_0 - 3}{4}} \big( 1+ o(1) \big)
  \end{equation}
 as $h \to 0$.
\end{itemize}
\end{lemma}

\begin{proof}
We have an explicit formula for $\Phi^\pm$ in terms of the Kummer function,
\[ \Phi^{\pm}(r,\theta) = C_{\pm} e^{im_{\pm} \theta} r^{m_{\pm}} \int_0^\infty e^{- \frac{B r^2}{4h}(1+2s)} (1+s)^{m_{\pm} - \delta_\pm} s^{\delta_\pm -1} ds,\]
where $\delta_\pm = \frac{1}{2} - \frac{\mu_{\pm}}{2 B h}$, and $C_{\pm} \neq 0$ is an $h$-dependent normalization constant. This formula can be derived as in \cite{FMR}, but one can also check directly that it is indeed a solution to the eigenvalue equation. In particular, $\Phi^\pm$ has the form \eqref{eq.Phipm2}, with
\begin{equation}\label{eq.defuh}
 u_h(r) = C_\pm R^{m_\pm} \Big( \frac{h}{BR^2}\Big)^{\frac{\Theta_0}{4}} e^{- \frac{BR^2}{4h}} \int_0^\infty e^{- \frac{Bs r^2}{2h} + m_\pm \ln( 1+s)} \frac{s^{\delta_\pm - 1}}{(1+s)^{\delta_\pm}} \dd s. 
 \end{equation}
The remaining of the proof consists in estimating $u_h$ and $C_\pm$ using the Laplace method. Note that the phase involves the two different frequencies $h$ and $\sqrt{h}$ via $m_{\pm}$.

1. To estimate $u_h(r)$, we start by the substitution $s = \frac{h}{BR^2} s'$ suggested by the Laplace method, which gives
\[ u_h(r) = C_\pm R^{m_\pm} \Big( \frac{h}{BR^2}\Big)^{\frac{\Theta_0}{4}  + \delta_\pm} e^{- \frac{BR^2}{4h}} \int_0^\infty e^{- \frac{s r^2}{2R^2} + m_\pm \ln \big( 1+ \frac{h}{BR^2} s \big)} s^{\delta_\pm - 1} \big(1+ \mathcal O ( h s) \big) \dd s. \]
 The involved phase satisfies
\begin{align*}
\frac{sr^2}{2R^2} - m_{\pm}(h) \ln \big( 1+ \frac{h}{BR^2} s \big) = \frac{ (r^2 - R^2) s}{2R^2} + \mathscr{O}(\sqrt{h}s).
\end{align*}
Therefore, the Laplace method gives
\begin{equation}\label{eq.uhr0}
 u_h(r) =  C_\pm R^{m_\pm} \Big( \frac{h}{BR^2}\Big)^{ \frac 12 - \frac{\Theta_0}{4}} e^{- \frac{BR^2}{4h}}  \Big( \frac{2 R^2}{ r^2-R^2} \Big)^{\delta_0} \Gamma ( \delta_0) \Big( 1 + o \Big( \frac{h^{\frac 12}}{r^2-R^2} \Big) \Big),
\end{equation}
where we also replaced $\delta_\pm$ by its asymptotic value $\delta_0 = \frac{1 - \Theta_0}{2}$.

2. Near $r=R$ we need finer estimates, and we consider for any $t>0$,
\[ u_h \Big(R + t \sqrt{\frac h B } \Big) =C_\pm R^{m_\pm} \Big( \frac{h}{BR^2}\Big)^{\frac{\Theta_0}{4}} e^{- \frac{BR^2}{4h}} \int_0^\infty e^{- \frac{Bs}{2 h} \big(R+t\sqrt{\frac hB} \big)^2 + m_{\pm} \ln(1+s)}  \frac{s^{\delta_{\pm}-1}}{(1+s)^{\delta_{\pm}}} \dd s.\]
This time, the Laplace method suggests the substitution $s = \sqrt{\frac{h}{BR^2}} s'$, and we obtain
\[ u_h \Big(R + t \sqrt{\frac h B } \Big)  = C_\pm R^{m_\pm} \Big( \frac{h}{BR^2}\Big)^{\frac{\Theta_0}{4} + \frac{\delta_\pm}{2}} e^{- \frac{BR^2}{4h}}  \int_0^\infty e^{- f(s,t,h)} s^{\delta_\pm - 1} \big(1+ \mathcal O ( \sqrt h s) \big) \dd s, \]
with phase
\begin{align*}
f(s,t,h) &=  \sqrt{\frac{BR^2}{h}} \frac s 2 \Big(1+ t \sqrt{\frac{h}{BR^2}} \Big)^2 - m_{\pm} \ln \Big(1+ s \sqrt{\frac{h}{BR^2}} \Big) \\
&= \sqrt{\frac{BR^2}{h}} \frac s 2 + ts + \frac{t^2 s}{2} \sqrt{\frac{h}{BR^2}} \\ &\qquad - \Big( \frac{BR^2}{2h} + \sqrt{\frac{\Theta_0 BR^2}{h}} + \mathscr{O}(1) \Big) \Big(s \sqrt{ \frac{h}{BR^2}} - \frac{h s^2}{2BR^2} + \mathscr{O}(h^{3/2}s^3) \Big) \\
&= ts + \frac{s^2}{4} -  \sqrt{\Theta_0} s + \mathscr{O}(\sqrt{h}s(1+s^2)).
\end{align*}
Thus, the Laplace method gives
\begin{equation}\label{eq.uhloc}
 u_h \Big(R + t \sqrt{\frac h B } \Big)  = C_\pm R^{m_\pm} \Big( \frac{h}{BR^2}\Big)^{\frac 14} e^{- \frac{BR^2}{4h}}  \int_0^\infty e^{- \big( \frac{s^2}{4} + ts - \sqrt{\Theta_0} s \big)} s^{\delta_0 - 1} \dd s  \big(1+ o(1) \big) .
 \end{equation}

3. We now estimate $C_\pm$, using that $\| \Phi^\pm \| =1$. In polar coordinates we deduce
\[1 =\Big( \frac{BR^2}{h} \Big)^{\frac{\Theta_0}{2}} \int_R^\infty \Big( \frac rR \Big)^{2 m_{\pm}} e^{- \frac{B}{2h}(r^2 - R^2)} |u_h(r)|^2 2 \pi r \dd r . \]
We use the substitution $r = R + t \sqrt{\frac h B}$, and the estimate \eqref{eq.uhloc} to obtain
\[  1 = C_\pm^2 R^{2m_\pm +1 }\Big( \frac{h}{BR^2}\Big)^{1-\frac{\Theta_0}{2}} \int_{0}^\infty e^{- \frac{B}{2h}\big( R + t\sqrt{\frac hB} \big)^2 + 2 m_\pm \ln \big( 1 + t\sqrt{\frac{h}{BR^2}} \big) } v(t)^2 2\pi R \big(1+o(1) \big) \dd t, \]
where 
\[ v(t) = \int_0^\infty e^{- \big( \frac{s^2}{4} + ts - \sqrt{\Theta_0} s \big)} s^{\delta_0 - 1} \dd s. \]
The involved phase satisfies
\begin{align*}
\frac{B}{2h}\Big(R + t\sqrt{\frac hB} \Big)^2 -2m_{\pm} \ln \Big(1+ t\sqrt{\frac{h}{BR^2}}\Big) = \frac{BR^2}{2h} + t^2 - 2 t \sqrt{\Theta_0} + \mathscr{O}(\sqrt{h}(1+t^3)).
\end{align*}
Thus, the Laplace method gives
\[ 1 =  C_\pm^2 R^{2 m_\pm+2 } \Big( \frac{h}{BR^2}\Big)^{1-\frac{\Theta_0}{2}} e^{- \frac{BR^2}{2h}} 2 \pi  \int_0^\infty e^{- \big( t^2 - 2t \sqrt{\Theta_0} \big)} v(t)^2 \dd t \big( 1+ o(1) \big).\]
Therefore we find
\begin{equation}\label{eq.Cpm}
C_\pm = K_0 R^{-m_\pm-1}  \Big( \frac{h}{BR^2}\Big)^{\frac{\Theta_0}{4}-\frac 12} e^{\frac{BR^2}{4h}} \big( 1+o(1) \big),
\end{equation}
where $K_0^{-2}= 2\pi \int_0^\infty e^{- \big( t^2 - 2t \sqrt{\Theta_0} \big)} v(t)^2 \dd t$. Now, we can insert \eqref{eq.Cpm} in \eqref{eq.uhr0} and \eqref{eq.uhloc} to find \eqref{eq.uhr} and \eqref{eq.uhR}.

4. Coming back to \eqref{eq.defuh} and differentiating we have
\[ \partial_r u_h(r) =  - C_\pm R^{m_\pm} \frac{ B r }{h} \Big( \frac{h}{BR^2}\Big)^{\frac{\Theta_0}{4}} e^{- \frac{BR^2}{4h}} \int_0^\infty e^{- \frac{Bs r^2}{2h} + m_\pm \ln( 1+s)} \frac{s^{\delta_\pm}}{(1+s)^{\delta_\pm}} \dd s.   \]
Therefore, similarly to \eqref{eq.uhloc} we find
\[ \partial_r u_h \Big( R + t\sqrt{\frac hB} \Big) = - C_\pm R^{m_\pm-1} \Big( \frac{h}{BR^2}\Big)^{-\frac 14} e^{- \frac{BR^2}{4h}}  \int_0^\infty e^{- \big( \frac{s^2}{4} + ts - \sqrt{\Theta_0} s \big)} s^{\delta_0} \dd s  \big(1+ o(1) \big) .\]
Note that the integral is equal to $-v'(t)$. It remains to use \eqref{eq.Cpm} to find \eqref{eq.uh'}.
\end{proof}

\section{Reduction \texorpdfstring{when $e_0 \neq \frac 12$}{(Case 1)}}\label{sec.reduction1}

In all of this section, we assume that $e_0 \in \big( - \frac 12, \frac 12 \big)$, and that $h$ is as small as we want along an $e_0$-sequence. In this case, Theorem \ref{thm.single} gives a spectral gap between the first two eigenvalues of the single-obstacle operator $\Lh^0$,
\begin{equation}
\mu_2 - \mu_1 \geq c_0 (1-2 |e_0|) h^2 + o(h^2).
\end{equation}
By Lemma \ref{lem.m-m+}, $\mu_1$ is the ground state energy of $\mathscr{L}_{h,m_-}$ with
\[m_- = \xi_h - e_0 + o(1) = \frac{BR^2}{2h} + \sqrt{\frac{\Theta_0 BR^2}{h}} + \mathcal C_2 - e_0 + o(1).\]
The ground state eigenfunction is $\Phi^-$.

We fix any discrete set of centers $V \subset \R^2$, with minimal distance
\[ L = \inf_{\substack{\alpha,\beta \in V \\ \alpha \neq \beta}} | \alpha - \beta |, \]
and we assume that the disks with centers $\alpha \in V$ and radius $R>0$ do not touch, i.e. $L>2R$. Then $\Lh^V$ denotes the magnetic-Neumann realization of $\Lh$ on \[ \Omega_V = \R^2 \setminus \left( \bigcup_{\alpha \in V} D(\alpha,R) \right).\]
We denote by $\Pi := \one_{(-\infty, \mu_1 + c_0 \frac{1-2|e_0|}{2} h^2)}(\Lh^V)$ the spectral projection of $\Lh^V$ corresponding to the interval $(-\infty, \mu_1 + c_0 \frac{1-2|e_0|}{2} h^2)$, and $\mathscr{E}_h = {\rm{Ran}}( \Pi)$.

\subsection{Agmon estimates}

It is well known that the eigenfunctions of $\Lh^V$ are exponentially localized near the boundary of the disks $D(\alpha,R)$. This follows from the so-called Agmon estimates. The following is an adaptation of the Agmon estimates in \cite[Theorem 8.2.4]{FH} to states $\Psi$ which are not necessarily eigenfunctions.

\begin{lemma}\label{lem.Agmon}
There exist $\kappa, C >0$ such that, if $h$ is small enough and if $\Psi \in L^2(\Omega_V)$ satisfies $\Psi = \one_{(-\infty, \frac{1+ \Theta_0}{2} B h]}(\Lh^V)\Psi$, then
\[ \int_{\Omega_V} e^{\frac{\kappa}{\sqrt{h}} \mathfrak{d}(x, \partial \Omega_V)} \Big( |\Psi |^2 + h^{-1} |(-ih \nabla - A) \Psi |^2 \Big) \dd x \leq C \| \Psi \|^2. \]
Here $\mathfrak{d}(x,\partial \Omega_V)$ is the Euclidean distance to the obstacles, and $A=(- \frac B2 x_2, \frac B2 x_1)$.
\end{lemma}

\begin{proof}
Define for $M > 0$, the regularized (at infinity) distance
$$
\mathfrak{d}_M =\min\{ \mathfrak{d},M\}.
$$
Let $\chi \in C^{\infty}({\mathbb R})$ be non-decreasing with $\chi(t) = 0$ for $t\leq 1$, $\chi(t)=1$ for $t\geq 2$, and let
$\chi_{\varepsilon}(x) = \chi(\mathfrak{d}/\varepsilon)$, for $\varepsilon >0$.
Define
$$
{\mathcal A} := \sup_{\{\Phi  = \one_{(-\infty, \frac{1+ \Theta_0}{2} B h]}(\Lh^V)\Phi \}}
\frac{ \| \chi_{\varepsilon} e^{\frac{\kappa}{\sqrt{h}} \mathfrak{d}_M} \Phi \|}{\|\Phi\|} \leq e^{\frac{\kappa}{\sqrt{h}} M}.
$$
We will prove that for $\varepsilon = \sqrt{h}/2$ and for sufficiently small values of $h$, ${\mathcal A}$ is bounded independently of $h$ and $M$.

Let $\Psi = \one_{(-\infty, \frac{1+ \Theta_0}{2} B h]}(\Lh^V)\Psi$ be normalized.
Consider the expression
\begin{align}\label{eq:agmon_Lowerquadratic}
\langle \chi_{\varepsilon} e^{\frac{\kappa}{\sqrt{h}} \mathfrak{d}_M} \Psi, \Lh^V \chi_{\varepsilon} e^{\frac{\kappa}{\sqrt{h}} \mathfrak{d}_M} \Psi \rangle \geq B h \| \chi_{\varepsilon} e^{\frac{\kappa}{\sqrt{h}} \mathfrak{d}_M} \Psi \|^2,
\end{align}
by the support properties of $\chi_{\varepsilon}$.
By the so-called IMS-localization formula (or equivalently: integration by parts), we find
\begin{align*}
\langle \chi_{\varepsilon} e^{\frac{\kappa}{\sqrt{h}} \mathfrak{d}_M} \Psi, &\Lh^V \chi_{\varepsilon} e^{\frac{\kappa}{\sqrt{h}} \mathfrak{d}_M} \Psi \rangle \\
&=
\Re \langle \chi_{\varepsilon}^2 e^{2\frac{\kappa}{\sqrt{h}} \mathfrak{d}_M} \Psi, \Lh^V \Psi \rangle
+ h^2 \langle \Psi, |\nabla (\chi_{\varepsilon} e^{\frac{\kappa}{\sqrt{h}} \mathfrak{d}_M} )|^2 \Psi \rangle
\end{align*}
We use that $\| \chi_\varepsilon e^{\frac{\kappa}{\sqrt h} \mathfrak{d}_M} \mathscr{L}_h^V \Psi \| \leq \mathcal{A} \| \mathscr{L}_h^V \Psi \|$ to deduce
\begin{multline}\label{eq:agmon_IMS}
\langle \chi_{\varepsilon} e^{\frac{\kappa}{\sqrt{h}} \mathfrak{d}_M} \Psi, \Lh^V \chi_{\varepsilon} e^{\frac{\kappa}{\sqrt{h}} \mathfrak{d}_M} \Psi \rangle \leq \frac{1+ \Theta_0}{2} B h{{\mathcal A}} \| \chi_{\varepsilon} e^{\frac{\kappa}{\sqrt{h}} \mathfrak{d}_M} \Psi \| 
+ 2 \kappa^2 h\| \chi_{\varepsilon} e^{\frac{\kappa}{\sqrt{h}} \mathfrak{d}_M} \Psi \| ^2  \\
+ 2 h^2 \varepsilon^{-2}\| \chi' \|_{\infty}^2 \int_{\{ 1 \leq \varepsilon^{-1} \mathfrak{d}(x) \leq 2\}} |e^{\frac{\kappa}{\sqrt{h}} \mathfrak{d}_M} \Psi|^2\,dx.
\end{multline}
We choose $\varepsilon = \sqrt{h}/2$ and get from combining \eqref{eq:agmon_Lowerquadratic} and \eqref{eq:agmon_IMS} and taking the supremum over all normalized $\Psi \in {\rm Ran} \one_{(-\infty, \frac{1+ \Theta_0}{2} B h]}(\Lh^V)$,
\begin{align}
Bh {\mathcal A}^2 \leq (\frac{1+ \Theta_0}{2} B h + 2 h \kappa^2) {\mathcal A}^2  + 8 h \| \chi' \|_{\infty}^2 e^{2\kappa}. 
\end{align}
If $\kappa^2 \leq \frac{1- \Theta_0}{4} B$, we therefore get
\begin{align}
{\mathcal A} \leq \left( \frac{1- \Theta_0}{4} B - \kappa^2\right)^{-1} 4 \| \chi' \|_{\infty}^2 e^{2\kappa}.
\end{align}
In particular, ${\mathcal A}$ is bounded independently of $h$ and $M$. Using \eqref{eq:agmon_IMS} again, we can also bound
\[ \| (-ih \nabla - A) \chi_\varepsilon e^{\frac{\kappa}{\sqrt h} \mathfrak{d}_M } \Psi \|^2 \leq Ch, \]
uniformly in $M$ and $h$. Therefore
\[ \int_{\Omega_V} \Big( \big| \chi_{\varepsilon} e^{\frac{\kappa}{\sqrt h} \mathfrak{d}_M } \Psi \big|^2 + h^{-1}\big| (-ih \nabla - A) \chi_{\varepsilon}  e^{\frac{\kappa}{\sqrt h} \mathfrak{d}_M } \Psi \big|^2 \Big) \dd x \leq C. \]
We can let $M \to \infty$ to change $\mathfrak d_M$ into $\mathfrak d$. Moreover, on the region where $\chi_\varepsilon \neq 1$ the exponential is bounded. We deduce that
\[ \int_{\Omega_V} \Big( \big| e^{\frac{\kappa}{\sqrt h} \mathfrak{d} } \Psi \big|^2 + h^{-1}\big| (-ih \nabla - A) e^{\frac{\kappa}{\sqrt h} \mathfrak{d}} \Psi \big|^2 \Big) \dd x \leq C,\]
and the result follows.
\end{proof}

In particular, using the Agmon estimates we can show that the spectrum of $\Lh^V$ in $(-\infty, \mu_1 + c_0 \frac{1-2|e_0|}{2} h^2)$ is exponentially close to $\mu_1$. 

\begin{lemma}\label{lem.spectralgap}
If $\lambda \in {\rm{sp}}( \Lh^V)$ is such that $\lambda < \mu_1 +  c_0 \frac{1-2|e_0|}{2} h^2$ then we have $|\lambda - \mu_1| \leq C_N h^N$ for all $N \geq 2$, where $C_N>0$ is independent of $\lambda$ and $h$. 
\end{lemma}

\begin{proof}
Let $\Psi$ be in the range of the spectral projector, $\Pi \Psi = \Psi$, and $(\theta_\alpha)_{\alpha \in V}$ a smooth partition of unity such that $\theta_\alpha$ is supported on a neighborhood of $D(\alpha,R)$ and $\sum_{\alpha} \theta_\alpha^2 =1$. Then $\theta_\alpha \Psi$ is in the domain of the single obstacle operator $\mathscr{L}^\alpha_h$ and by the spectral theorem
\[ \sum_{\alpha \in V} \| (\Lh^\alpha - \lambda) \theta_\alpha \Psi \|^2 \geq |\mu_1 - \lambda|^2 \sum_{\alpha \in V} \| \theta_\alpha \Psi \|^2 = |\mu_1 - \lambda |^2 \|\Psi\|^2, \]
because $\mu_1$ is the closest eigenvalue of $\Lh^\alpha$ to $\lambda$ --this is our assumption on $\lambda$. Since $\theta_\alpha \Psi$ is also in the domain of $\Lh^V$,
\[ \sum_{\alpha \in V}  \| (\Lh^\alpha - \lambda) \theta_\alpha \Psi \|^2 \leq \sum_{\alpha \in V} \Big( \| \big[ \Lh^V, \theta_\alpha \big] \Psi \|^2 + \|\theta_\alpha (\Lh^V - \lambda) \Psi \|^2 \Big) \]
The commutator is supported away from the disks $D(\alpha,R)$. Therefore, using Agmon estimates from Lemma \ref{lem.Agmon} for $\Psi$, we deduce that it is exponentially small,
\begin{align*}
 \sum_{\alpha \in V}  \| (\Lh^\alpha - \lambda) \theta_\alpha \Psi \|^2 &\leq C_N h^N \| \Psi \|^2 + \|(\Lh^V - \lambda) \Psi \|^2
\end{align*}
We combine this with the lower bound to obtain 
\begin{equation}
|\mu_1 - \lambda|^2 \leq C_N h^N + \inf_{\Psi \neq 0} \frac{\| ( \mathscr{L}_h^V - \lambda) \Pi \Psi \|^2}{\| \Pi \Psi \|^2}
\end{equation}
Since the operator $\Lh^V$ restricted to $\rm{Ran} (\Pi)$ is selfadjoint, and $\lambda$ in its spectrum, the above infimum is vanishing. This concludes the proof.
\end{proof}

\subsection{Approximate eigenfunctions \texorpdfstring{$\Psi_\alpha$ and $g_{\alpha}$}{}}

$\Phi^-$ is the ground state of $\Lh^0$. We can translate and cut-off $\Phi^-$ to construct quasimodes for $\Lh^V$. For any $\eta >0$, let $\chi_\eta$ be a smooth radial cutoff function satisfying
\begin{equation}\label{def.cutoff}
\chi_\eta(x) =
\begin{cases}
1 \quad \text{if $|x| < L-R-2 \eta$}\\
0 \quad \text{if $|x| > L-R- \eta$}.
\end{cases}
\end{equation}
We then define, for $\alpha \in V$,
\begin{equation}\label{def.Psialpha}
\Psi_\alpha = \tau_\alpha^B \big( \chi_\eta \Phi^- \big), \qquad g_\alpha = \Pi \Psi_\alpha.
\end{equation}

In this section, we provide useful estimates on $g_\alpha - \Psi_\alpha$. First of all, note that the states $\Psi_\alpha$ are quasimodes for $\Lh^V$, approximately solving the eigenvalue equation.

\begin{lemma}\label{lem.Psiv}
If $\eta$ is small enough then, with the function $d$ defined in \eqref{eq:d_function},
\[ \| (\Lh^V - \mu_1) \Psi_\alpha \| = \mathscr{O}\big( e^{- \frac{d(L-R-3\eta)}{h}} \big), \quad \| \Psi_\alpha\| = 1 + \mathscr{O}\big( e^{- \frac{d(L-R-3\eta)}{h}} \big), \]
and the estimates are uniform with respect to $\alpha \in V$.
\end{lemma}

\begin{proof}
Since $L= \inf_{\substack{\alpha,\beta \in V \\ \alpha \neq \beta}} |\alpha-\beta|$ and by definition of $\chi_\eta$, the function $\Psi_\alpha$ vanishes on an $\eta$-neighborhood of each disk $D(\beta,R)$ for $\beta \neq \alpha$. Since it also satisfies the Neumann boundary condition on $D(\alpha,R)$, it is in the domain of $\Lh^V$. Moreover $\tau_\alpha^B$ commutes with $\Lh$ so we have
\begin{align*}
\|  (\Lh^V - \mu_1) \Psi_\alpha \| = \|  (\Lh^0 - \mu_1) \chi_\eta \Phi^- \| = \| \big[ \mathscr{L}_h^0 , \chi_\eta \big] \Phi^- \|,
\end{align*}
where we used that $\Phi^-$ was an eigenfunction for $\Lh^{0}$ with eigenvalue $\mu_1$. Now the commutator is supported at distance $|x|> L-R-2\eta$, so we can use the decay properties of $\Phi^-$ and its gradient from Lemma \ref{lem.decay} to get
\begin{align*}
\|  (\Lh^V - \mu_1) \Psi_\alpha \| = \mathscr{O}\big( e^{- \frac{d(L-R-3\eta)}{h}} \big),
\end{align*}
where we replaced $2 \eta$ by $3\eta$ inside the exponential to control the front powers of $h^{-1}$. We estimate the norm similarly, using that $\| \tau_\alpha^B \Phi^- \| =1$ and
\[ \| \Psi_\alpha - \tau_\alpha^B \Phi^- \| = \| (\chi_\eta - 1) \Phi^{-} \| = \mathscr{O}\big( e^{- \frac{d(L-R-3\eta)}{h}}  \big). \]
Note that in these estimates the translation by $\alpha$ does not change the norm, and therefore they are independent of $\alpha$.
\end{proof}

It follows that $g_\alpha$ is a small perturbation of $\Psi_\alpha$, and that these vectors form a basis of $\mathscr E_h$.

\begin{lemma}\label{lem.gv}
For all $\alpha \in V$, the state $\Psi_\alpha$ and its projection $g_\alpha = \Pi \Psi_\alpha$ on $\mathscr E_h$ satisfy, with the function $d$ defined in \eqref{eq:d_function},
\begin{equation*}
\| g_\alpha - \Psi_\alpha \| = \mathscr O \big( e^{-\frac{d(L-R-4\eta)}{h}} \big), \quad \langle \Lh^V (g_\alpha - \Psi_\alpha), g_\alpha - \Psi_\alpha \rangle =  \mathscr O \big( e^{-\frac{2d(L-R-4\eta)}{h}} \big).
\end{equation*}
Moreover, the vectors $(g_\alpha)_{\alpha \in V}$ span a dense subspace of $\mathscr E_h$.
\end{lemma}

\begin{proof}
By definition of the spectral projector $\Pi$ we have
\[\| (\Lh^V - \mu_1)(I- \Pi) \Psi_\alpha \| \geq c_0 \frac{(1-2|e_0|) h^2}{2} \| (I- \Pi) \Psi_\alpha \|.\]
Moreover, the projector $(I-\Pi)$ commutes with $\Lh^V$ so we have the upper bound
\[\| (\Lh^V - \mu_1)(I- \Pi) \Psi_\alpha \| \leq \| (\Lh^V - \mu_1) \Psi_\alpha \| = 
\| (\Lh^{\alpha} - \mu_1) \Psi_\alpha \| =
\mathscr{O}\big( e^{- \frac{d(L-R-3\eta)}{h}} \big), \]
where the last estimate is given by Lemma \ref{lem.Psiv}. Combining these two bounds we obtain
\[ \| (I- \Pi) \Psi_\alpha \| \leq C h^{-2} e^{- \frac{d(L-R-3\eta)}{h}},  \]
and this proves the bound on $\| g_\alpha - \Psi_\alpha \|$. To get the quadratic form bound we use a Cauchy-Schwarz inequality,
\[\langle \Lh^V (I - \Pi) \Psi_\alpha, (I - \Pi) \Psi_\alpha \rangle \leq \| \Lh^V (I - \Pi) \Psi_\alpha \| \| (I - \Pi) \Psi_\alpha \| \leq C h^{-4}e^{- \frac{2d(L-R-3\eta)}{h}},\]
where we used Lemma \ref{lem.Psiv} again. Note that we can absorb the $h^{-4}$ by changing $3\eta$ into $4\eta$ in the exponential. The bounds are uniform with respect to $\alpha$, since the ones of Lemma~\ref{lem.Psiv} are. 

Let us now prove that the vectors $(g_\alpha)_{\alpha \in V}$ span a dense subspace of $\mathscr E_h$. Note that, when $V$ is finite $(g_\alpha)_{\alpha \in V}$ is obviously a basis for cardinality reasons, but in the general case an additional argument is needed. Assume that $\Psi \in \mathscr{E}_h$ is orthogonal to all $g_\alpha$'s, and we have to prove that this implies $\Psi=0$. First note that $\Psi$ is also orthogonal to $\Psi_\alpha$ since
\begin{equation}
\langle \Psi , \Psi_\alpha \rangle = \langle \Pi \Psi, \Psi_\alpha \rangle = \langle \Psi, g_\alpha \rangle = 0.
\end{equation}
We introduce a family $( \theta_\alpha)_{\alpha \in V}$ of cutoffs such that $\theta_\alpha$ is localized on a small but fixed neighborhood of $D(\alpha,R)$, and satisfying
\[\sum_{\alpha \in V} \theta_\alpha^2 = 1.\]
We then use the IMS localization formula,
\begin{align*}
\langle \Lh^V \Psi, \Psi \rangle = \sum_{\alpha \in V} \langle \Lh^V \theta_\alpha \Psi, \theta_\alpha \Psi \rangle - h^2 \sum_{\alpha \in V} \| ( \nabla \theta_\alpha ) \Psi \|^2.
\end{align*}
Using the Agmon estimates from Lemma \ref{lem.Agmon}, we know that $\Psi$ is exponentially localized near the disks. Since $\sum_\alpha |\nabla \theta_\alpha|^2$ is supported away from the disks we deduce
\begin{align*}
\langle \Lh^V \Psi, \Psi \rangle \geq \sum_{\alpha \in V} \langle \Lh^V \theta_\alpha \Psi, \theta_\alpha \Psi \rangle - Ch^N \| \Psi \|^2,
\end{align*}
where we can choose $N$ arbitrarily large (here $N=3$ is enough). Note that $\theta_\alpha \Psi$ vanishes on a neighborhood of all disks $D(\beta,R)$ for $\beta \neq \alpha$. Therefore $\Lh^V \theta_\alpha \Psi = \Lh^\alpha \theta_\alpha \Psi$, where we recall that $\Lh^\alpha$ is the operator with a single obstacle $D(\alpha,R)$, i.e. $\Lh^\alpha = \tau_{\alpha}^B \Lh^0 \tau_{-\alpha}^{B}$. Thus
\begin{align}\label{eq.boundA}
\langle \Lh^V \Psi, \Psi \rangle \geq \sum_{\alpha \in V} \langle \Lh^0 \tau_{-\alpha}^B \theta_\alpha \Psi, \tau_{-\alpha}^B \theta_\alpha \Psi \rangle - Ch^N \| \Psi \|^2.
\end{align}
We now introduce the projection $\Pi_0$ on the ground state of $\Lh^0$, which is given by
\[\Pi_0 \tau_{-\alpha}^B \theta_\alpha \Psi = \langle \tau_{-\alpha}^B \theta_\alpha \Psi , \Phi^- \rangle \Phi^-.\]
Therefore
\begin{align*}
\langle \Lh^0 \tau_{-\alpha}^B \theta_\alpha \Psi, \tau_{-\alpha}^B \chi_\alpha \Psi \rangle &= \mu_1 | \langle \tau_{-\alpha}^B \theta_\alpha \Psi , \Phi^- \rangle |^2 +  \langle \Lh^0 (I-\Pi_0) \tau_{-\alpha}^B \theta_\alpha \Psi, (I-\Pi_0) \tau_{-\alpha}^B \theta_\alpha \Psi \rangle\\
&\geq \mu_1 | \langle \tau_{-\alpha}^B \theta_\alpha \Psi , \Phi^- \rangle |^2 + (\mu_1 + c_0(1-2|e_0|) h^2) \|  (I-\Pi_0) \tau_{-\alpha}^B \theta_\alpha \Psi \|^2,
\end{align*}
where we used that the second eigenvalue of $\Lh^0$ has a gap of size $c_0 (1-2|e_0|) h^2$ with $\mu_1$ (Theorem \ref{thm.single}). We reconstruct the $L^2$ norm of $ \tau_{-\alpha}^B \theta_\alpha \Psi$  and use that $\tau_{-\alpha}^B$ is an isometry to get
\begin{align}\label{eq.2704}
&\langle \Lh^0 \tau_{-\alpha}^B \theta_\alpha \Psi, \tau_{-\alpha}^B \chi_\alpha \Psi \rangle \nonumber \\
&\geq \big ( \mu_1 + c_0(1-2|e_0|) h^2 \big) \| \theta_\alpha \Psi \|^2 - c_0(1-2|e_0|) h^2 | \langle \tau_{-\alpha}^B \theta_\alpha \Psi , \Phi^- \rangle |^2.
\end{align}
The last term can be estimated as follows. First introduce the cutoff $\chi_\eta$ defining $\Psi_\alpha$,
\begin{align*}
 \langle \tau_{-\alpha}^B \theta_\alpha \Psi , \Phi^- \rangle = \langle \theta_\alpha \Psi , \tau_\alpha^B \chi_\eta \Phi^- \rangle + \langle  \theta_\alpha \Psi, \tau_\alpha^B (1-\chi_\eta) \Phi^- \rangle.
\end{align*}
Using the decay of $\Phi^-$ and the cutoff we see that the second term is exponentially small. In the first one we use the orthogonality $\langle \Psi, \Psi_\alpha \rangle =0$ and we find
\begin{align*}
 \langle \tau_{-\alpha}^B \theta_\alpha \Psi , \Phi^- \rangle &\leq \langle (\theta_\alpha-1) \Psi , \tau_\alpha^B \chi_\eta \Phi^- \rangle +Ch^N \| \theta_\alpha \Psi \|\\
&\leq \| \chi_\eta( . -\alpha) \big( \theta_\alpha -1 \big) \Psi \| +Ch^N \| \theta_\alpha \Psi \|.
\end{align*}
We insert this last inequality in \eqref{eq.2704} and equation \eqref{eq.boundA} now reads
\begin{equation}
\langle \Lh^V \Psi, \Psi \rangle \geq \big( \mu_1 + c_0(1-2|e_0|) h^2 - Ch^N \big) \| \Psi \|^2 - C h^2 \sum_{\alpha \in V} \| \chi_\eta( . -\alpha) \big( \theta_\alpha -1 \big) \Psi \|^2.
\end{equation}
Since the function $\sum_\alpha \chi_\eta(.-\alpha)^2 \big( \theta_\alpha -1 \big)^2$ is supported away from the obstacles, we can use Agmon estimates again to see that the last term is exponentially small, and therefore
\begin{equation}
\langle \Lh^V \Psi, \Psi \rangle \geq \big( \mu_1 + c_0(1-2|e_0|) h^2 - Ch^N \big) \| \Psi \|^2.
\end{equation}
But $\Psi \in \mathscr E_h$, meaning that $\langle \Lh^V \Psi, \Psi \rangle \leq ( \mu_1 + c_0 \frac{1-2|e_0|}{2}h^2) \| \Psi \|^2$, and therefore we must have $\Psi = 0$. This proves that $(g_\alpha)$ generates a dense subspace of $\mathscr{E}_h$.
\end{proof}

In the case when the set of disks $V$ is infinite, we need a little more information on the decay of $g_\alpha - \Psi_\alpha$, which we give in the following lemma. This result is very weak, since we only get decay when $\langle x - \alpha \rangle \gg h^{-1}$. However, this is enough to obtain decay of coefficients of our infinite interaction matrix away from the diagonal, in the proof of Theorem \ref{thm.Gram}.

\begin{lemma}\label{lem.weighted}
For all $\alpha \in V$ we have
\[ \| e^{ h \langle x - \alpha \rangle} ( g_\alpha - \Psi_\alpha) \| = \mathscr{O}( e^{-\frac{d(L-R-4\eta)}{h}} ),  \qquad \| e^{ h \langle x - \alpha \rangle} \Lh^V ( g_\alpha - \Psi_\alpha) \| = \mathscr{O}( e^{-\frac{d(L-R-4\eta)}{h}} ), \]
where $\langle x \rangle = (1+ |x|^2)^{1/2}$.
\end{lemma}

\begin{proof}
We recall that $\Psi_\alpha$, defined in \eqref{def.Psialpha}, is an approximate eigenfunction for $\Lh^V$. We define the remainder
\[ r_\alpha = (\Lh^V - \mu_1) \Psi_\alpha,\]
which satisfies
\begin{equation}\label{eq.approx1005}
(\Lh^V - z)^{-1} \Psi_\alpha = \frac{1}{\mu_1 - z} \Psi_\alpha - \frac{1}{\mu_1-z} (\Lh^V-z)^{-1} r_{\alpha}, \quad \forall z \notin {\rm{sp}}( \Lh^V).
\end{equation}
Also note that the spectral projection $g_\alpha =\Pi \Psi_\alpha$ can be written as the path-integral
\begin{equation}
g_\alpha = - \frac{1}{2\pi i} \int_{\Gamma} (\Lh^V - z)^{-1} \Psi_\alpha \dd z,
\end{equation}
where $\Gamma$ is a (clock-wise) circle of center $\mu_1$ and radius $c_0 \frac{1-2|e_0|}{2}h^2$. Therefore using \eqref{eq.approx1005}, we have
\begin{equation} \label{eq.integralg}
 g_\alpha - \Psi_\alpha = \frac{1}{2\pi i} \int_{\Gamma} \frac{1}{\mu_1 -z} (\Lh^V - z)^{-1} r_\alpha \dd z, 
 \end{equation}
and
\begin{equation} \label{eq.bound1005}
\| e^{ h \langle x- \alpha \rangle} (g_\alpha - \Psi_\alpha) \| \leq C  \sup_{z \in \Gamma} \| e^{ h \langle x-\alpha \rangle} (\Lh^V - z)^{-1} r_\alpha \| 
\end{equation}
In order to bound the right-hand side, we first consider the weight $\phi(x) = \langle x - \alpha \rangle \rho(x/M)$, where $\rho \in C^\infty_0(\R^2)$ is equal to $1$ on $D(0,1)$. We will prove estimates uniform in $M\geq 1$ and then let $M \to \infty$. 

For $z \in \Gamma$, we use the resolvent commutation formula to find
\begin{align}\nonumber
 \| e^{ h \phi} (\Lh^V - z)^{-1} r_\alpha \| &\leq \| (\Lh^V-z)^{-1} e^{ h \phi} r_\alpha  \| + \| (\Lh^V -z)^{-1} \big[ \Lh^V, e^{ h \phi}\big] (\Lh^V - z)^{-1} r_\alpha \| \\
&\leq C h^{-2} \| e^{ h \phi} r_\alpha  \| + C h^{-2} \|\big[ \Lh^V, e^{ h \phi}\big] (\Lh^V - z)^{-1} r_\alpha \|, \label{eq.est1005}
\end{align}
where we used the spectral theorem and Lemma \ref{lem.spectralgap} to bound the norm of the resolvent. Let us now investigate the commutator, which is explicitly given by
\[ \big[ \Lh^V, e^{ h \phi} \big] =  -e^{ h \phi}\big( 2i h^2 \nabla \phi \cdot (-ih\nabla - A) + h^3 \Delta \phi + h^4 |\nabla \phi |^2 \big), \]
where $A=(-\frac B2 x_2, \frac B2 x_1)$ is the vector potential.
Since $\nabla \phi$ and $\Delta \phi$ are bounded (uniformly with respect to $\alpha$) we deduce
\begin{align}\nonumber
\| \big[ \Lh^V, e^{ h \phi} \big] (\Lh^V - z)^{-1} r_\alpha \|&\leq C  h^2 \| e^{ h \phi} (-ih\nabla - A) (\Lh-z)^{-1} r_\alpha \|  \\ &\qquad + C  h^3 \| e^{ h \phi} (\Lh-z)^{-1} r_\alpha \|. \label{eq.commute1005}
\end{align}
The first term in \eqref{eq.commute1005} we bound as follows using the notation $u = (\Lh^V - z)^{-1} r_\alpha$.
We integrate by parts, and since $u \in {\rm{Dom}}(\Lh^V)$ has Neumann boundary conditions, we find
\begin{align*}
\| e^{ h \phi} (-ih \nabla - A) u \|^2 &= \langle (-ih \nabla - A) e^{2 h \phi} (-ih\nabla - A) u, u \rangle \\
&\leq | \langle e^{2  h \phi} \Lh u, u \rangle | + 2| \langle  h^2 \nabla \phi e^{ h \phi} (-ih\nabla - A)u,e^{ h \phi} u \rangle |\\
&\leq \| e^{ h \phi} \Lh u \| \| e^{ h \phi} u \| + C  h^2 \|e^{ h \phi} (-ih\nabla - A)u \| \| e^{ h \phi} u \|.
\end{align*}
In the last term, we use the Cauchy inequality and compare with the left-hand side to get
\begin{align*}
\| e^{ h \phi} (-ih \nabla - A) u \|^2 &\leq  C \| e^{ h \phi} \Lh u \| \| e^{ h \phi} u \| + C   h^4\| e^{ h \phi} u \|^2 \\
&\leq C\| e^{ h \phi} (\Lh^V - z) u \| \| e^{ h \phi} u \| + C h \| e^{ h \phi} u \|^2\\
&\leq C h^{-1}  \| e^{ h \phi} (\Lh^V - z) u \|^2  + C  h \| e^{ h \phi} u \|^2  .
\end{align*}
We recall that $u = (\Lh^V - z)^{-1} r_\alpha$ and insert this last bound in \eqref{eq.commute1005},
\begin{align}\nonumber
\| \big[ \Lh^V, e^{ h \phi} \big] (\Lh^V - z)^{-1} r_\alpha \|&\leq C  h^{3/2} \| e^{ h \phi} r_\alpha \|  + C  h^{5/2} \| e^{ h \phi} (\Lh-z)^{-1} r_\alpha \|. 
\end{align}
We now insert in \eqref{eq.est1005} and obtain
\[ \| e^{ h \phi} (\Lh^V - z)^{-1} r_\alpha \| \leq  C h^{-2} \| e^{ h \phi} r_\alpha  \| + C  h^{1/2} \| e^{ h \phi} (\Lh^V - z)^{-1} r_\alpha \|. \]
We substract the last term from the left-hand side and get
\[  \| e^{ h \phi} (\Lh^V - z)^{-1} r_\alpha \| \leq  C h^{-2} \| e^{ h \phi} r_\alpha  \|. \]
Then, we let $M \to \infty$, so that $\phi \to \langle x-\alpha \rangle$, and with \eqref{eq.bound1005} this gives
\[ \| e^{ h \langle x- \alpha \rangle} (g_\alpha - \Psi_\alpha) \| \leq Ch^{-2} \| e^{ h \langle x-\alpha \rangle} r_\alpha  \|. \]
Finally, we can estimate the norm of $e^{ h\langle x-\alpha \rangle} r_\alpha$ as in Lemma \ref{lem.Psiv}. Indeed, this estimate follows from the decay of the groundstate $\Phi^-$ of $\Lh^0$, which is much faster than $e^{h\langle x-\alpha \rangle}$, and the proof is thus identical. We obtain the first stated result. For the second one, we use \eqref{eq.integralg} again and apply $\Lh^V$ to find
\begin{align*}
\| e^{ h \langle x - \alpha \rangle} \Lh^V ( g_\alpha - \Psi_\alpha) \| \leq C  \| e^{h \langle x - \alpha \rangle} r_\alpha \| + C \sup_{z \in \Gamma} |z| \| e^{h \langle x - \alpha \rangle} (\Lh^V - z)^{-1} r_\alpha \|,
\end{align*}
and these two terms we estimate as above.
\end{proof}

\subsection{The Gram matrix}

In the previous section, we constructed a basis $(g_\alpha)_{\alpha \in V}$ of the spectral subspace $\mathscr{E}_h$. Before estimating the spectrum of $\Lh^V$ restricted to $\mathscr{E}_h$, we orthonormalize the basis $(g_\alpha)_{\alpha \in V}$, thus providing a Hilbertian basis $(e_\alpha)_{\alpha \in V}$ of $\mathscr{E}_h$. This is done using the Gram matrix $\mathbf G$. For any $\varepsilon >0$ we define the error order
\begin{align}
\mathcal E_{\varepsilon,L} = e^{- \frac{2S_h(L)}{h}} + e^{-\frac{2d(L-R-4\eta)}{h}}+ e^{-\frac{S_h(L(1+\varepsilon))}{h}} + e^{-\frac{d(L-R-4\eta) + d(\varepsilon L + R + 2 \eta)}{h}}, 
\end{align}
where $S_h$ was defined in \eqref{eq:Sh}.
This error is shown to be small enough in Lemma \ref{lem.errors}.

\begin{theorem}\label{thm.Gram}
Let $\varepsilon >0$ and assume $h>0$ is small enough along an $e_0$-sequence with $e_0 \in (-\frac 12, \frac 12)$. For all $p \in [1, \infty]$, the infinite matrix $\G = (\langle g_\alpha, g_\beta \rangle)_{\alpha, \beta \in V}$ defines a bounded operator on $\ell^p(V)$. Moreover,
\begin{enumerate}
\item We have the following estimate in $\ell^p$ operator norm topology
\[\G = \mathbf I + \T + \mathscr{O}(\mathcal E_{\varepsilon,L}), \quad {\text{in}}\quad \mathcal{L}(\ell^p(V)),\]
where $\T_{\alpha \beta} = \langle g_\alpha, g_\beta \rangle \one_{L \leq |\alpha - \beta| \leq L(1 + \varepsilon )}$. 
\item $\G^{-1}$ is a bounded operator on $\ell^p(V)$ and has a square root which satisfies
\[ \G^{- \frac 1 2} = \mathbf I - \frac 1 2 \T + \mathscr{O}(\mathcal E_{\varepsilon,L}), \quad {\text{in}}\quad \mathcal{L}(\ell^p(V)). \]
\item The vectors
\[e_\alpha = \sum_{\beta \in V} (\G^{- \frac 12})_{\alpha \beta} g_\beta, \qquad \alpha \in V,\]
define a Hilbert basis of $\mathscr{E}_h$.
\item The restriction of $\Lh^V$ to $\mathscr{E}_h$ is unitarily equivalent to $\M = ( \langle \Lh^V e_\alpha, e_\beta \rangle)_{\alpha, \beta \in V}$ as an operator on $\ell^2(V)$, and
\[ \M = \mu_1 \mathbf I + \mathbf W +\mathscr{O}(\mathcal E_{\varepsilon,L}),\]
where $\mathbf W_{\alpha \beta} = \langle (\Lh - \mu_1) \Psi_\alpha, \Psi_\beta \rangle \one_{L \leq |\alpha - \beta | \leq L(1 + \varepsilon)}$.
\end{enumerate}
\end{theorem}

\begin{proof}
Some of the technical estimates on the interaction coefficients required in the proof of Theorem~\ref{thm.Gram} are postponed to Appendix~\ref{appendix:IC} below.

1. We prove the result for $p=1$ and $p=\infty$. The general statement then follows by the interpolation inequality,
\[ \| \mathbf A \|_p \leq \| \mathbf A \|_1^{\frac 1 p} \| \mathbf A \|_{\infty}^{1-\frac 1 p} ,\]
where here $\| \mathbf A \|_p$ is the operator norm $\ell^p(V) \to \ell^p(V)$. Note that for hermitian matrices like $\G$, the $\ell^1$-norm and $\ell^{\infty}$-norm are equal and given by
\[ \| \mathbf A \|_1 = \| \mathbf A \|_{\infty} = \sup_{\alpha \in V} \sum_{\beta \in V} |\mathbf A_{\alpha \beta}| .\]
We can write $\G$ as follows, decomposing according to $|\alpha - \beta|$,
\[ \G = \widetilde{\mathbf I} + \mathbf T + \widetilde{\mathbf R}, \]
where 
\begin{align*}
\widetilde{\mathbf I}_{\alpha \beta} = \| g_\alpha \|^2 \mathbf{1}_{\alpha = \beta}, \quad
{\mathbf T}_{\alpha \beta} = \langle g_\alpha, g_\beta \rangle \mathbf{1}_{L \leq |\alpha - \beta| \leq L(1 + \varepsilon)}, \quad
\widetilde{\mathbf R}_{\alpha\beta} = \langle g_\alpha, g_\beta \rangle \mathbf{1}_{|\alpha - \beta| >L(1 + \varepsilon)}.
\end{align*}
We first estimate the norm of $\widetilde{\mathbf R}$,
\begin{align*}
\| \widetilde{\mathbf R} \|_1 = \sup_{\alpha \in V} \sum_{\substack{\beta \in V \\ |\alpha - \beta|>L(1 + \varepsilon)}} |\langle g_\alpha, g_\beta \rangle|.
\end{align*}
By the Pythagorean Theorem we have $\langle g_\alpha, g_\beta \rangle = \langle \Psi_\alpha, \Psi_\beta \rangle + \langle g_\alpha - \Psi_\alpha, g_\beta - \Psi_\beta \rangle$. Remember that $\Psi_\alpha$ is supported on $D(\alpha, L-R-\eta)$. Therefore, given $\alpha \in V$, the scalar product $\langle \Psi_\alpha, \Psi_\beta \rangle$ vanishes for all but finitely many $\beta$'s. For those we can use Lemma~\ref{lem.scalarproduct} below, and we deduce
\begin{align}\label{est.R1}
\| \widetilde{\mathbf R} \|_1 \leq \sup_{\alpha \in V} \sum_{\substack{\beta \in V \\ |\alpha - \beta|>L(1 + \varepsilon)}} | \langle g_\alpha - \Psi_\alpha, g_\beta - \Psi_\beta \rangle | + \mathscr{O} \big(e^{-\frac{S_h(L(1+ \varepsilon))}{h}} + e^{-\frac{d(L-R-3\eta) + d(\varepsilon L + R + 2 \eta)}{h}} \big)
\end{align}
To estimate the remaining sum, we use that $\langle \alpha - \beta \rangle \leq \langle \alpha - x \rangle + \langle \beta - x \rangle$ to insert a decaying weight,
\begin{align*}
| \langle g_\alpha - \Psi_\alpha, g_\beta - \Psi_\beta \rangle | \leq e^{-h \langle \alpha - \beta \rangle} \| e^{h \langle x - \alpha \rangle }( g_\alpha - \Psi_\alpha) \| \|e^{h \langle x - \beta \rangle }( g_\beta - \Psi_\beta) \|,
\end{align*}
and then we use Lemma \ref{lem.weighted} to get
\begin{align*}
| \langle g_\alpha - \Psi_\alpha, g_\beta - \Psi_\beta \rangle | \leq C e^{-h\langle \alpha - \beta \rangle} e^{- \frac{2 d(L-R-4\eta)}{h}}.
\end{align*}
The sum over $\beta$ is convergent and \eqref{est.R1} becomes
\begin{equation}\label{eq.bdR}
\| \widetilde{\mathbf R }\|_1 \leq C \big(e^{- \frac{2 d(L-R-4\eta)}{h}} + e^{-\frac{S_h(L(1+\varepsilon))}{h}} + e^{-\frac{d(L-R-4\eta) + d(\varepsilon L + R + 2 \eta)}{h}} \big) .
\end{equation}
Similarly we can estimate
\begin{align*}
\| \mathbf I - \widetilde{\mathbf I} \|_1 &\leq C \sup_{\alpha \in V} | \| g_\alpha \|^2 -1 | \leq C  e^{- \frac{2d(L-R- 4\eta)}{h}},
\end{align*}
using Lemma \ref{lem.gv}. This proves the first item.

2. Note that the previous discussion also gives, for $p \in [1,\infty]$,
\begin{align*}
\| \mathbf G - \mathbf I \|_p \leq \| \mathbf T \|_p + \mathscr{O}(e^{- \frac{2 d(L-R-4\eta)}{h}} + e^{-\frac{S_h(L(1+\varepsilon))}{h}} + e^{-\frac{d(L-R-4\eta) + d(\varepsilon L + R + 2 \eta)}{h}}),
\end{align*}
and using Lemma~\ref{lem.gv}, as well as Lemma~\ref{lem.scalarproduct} from Appendix~\ref{appendix:IC}, we get
\begin{align*}
\| \mathbf T \|_p \leq \| \mathbf T \|_1 &\leq C \sup_{L \leq |\alpha - \beta| \leq L(1 + \varepsilon)} | \langle \Psi_\alpha, \Psi_\beta \rangle | + C e^{-\frac{2d(L-R-4\eta)}{h}}  \\ &\leq C \big(e^{-\frac{S_h(L)}{h}} + e^{-\frac{d(L-R-3\eta)}{h}} \big).
\end{align*}
Therefore $\G$ is a perturbation of the identity, and as such it is invertible. Moreover, $\mathbf G^{-1/2}$ exists -- as sum of a convergent Taylor series -- and the first terms of the series give
\begin{align*}
\| \G^{- 1/2} - \mathbf I + \frac 1 2 \mathbf T \|_p \leq C \| \T \|_p^2 + C \big(e^{- \frac{2 d(L-R-4\eta)}{h}} + e^{-\frac{S_h(L(1+\varepsilon))}{h}} + e^{-\frac{d(L-R-4\eta) + d(\varepsilon L + R + 2 \eta)}{h}} \big),
\end{align*}
which proves the second item.

3. The series defining $e_\alpha$ is convergent in $\mathscr{E}_h$ because $\G^{- \frac 1 2}$ is a bounded operator on $\ell^{1}(V)$ and
\[ \sum_{\beta \in V}  | (\mathbf G^{-\frac 12})_{\alpha \beta}| \|g_\beta \| \leq C \sum_{\beta} | (\mathbf G^{-\frac 12})_{\alpha \beta}| \leq C \| \mathbf G \|_1. \]
Moreover, item (2) for $p=1$ yields
\begin{equation}
 e_\alpha = g_\alpha - \frac 1 2 \sum_{\substack{\beta \in V \\ L \leq |\alpha - \beta |\leq L(1 + \varepsilon)}} \langle g_\alpha, g_\beta \rangle g_\beta + \mathscr{O}(\mathcal E_{\varepsilon, L}),
\end{equation}
uniformly with respect ot $\alpha$. In particular, $e_\alpha$ is a perturbation of $g_\alpha$. Using Lemma \ref{lem.gv}, we deduce that it spans a dense subspace of $\mathscr{E}_h$. It is also an orthonormal family, since
\begin{align*}
\langle e_\alpha, e_{\alpha'} \rangle &= \sum_{\beta, \beta' \in V} \G^{- \frac 1 2}_{\alpha \beta}  \langle g_\beta, g_{\beta'} \rangle \overline{\G^{- \frac 1 2}_{\alpha' \beta'}} = (\G^{- \frac 12} \G \G^{- \frac 12})_{\alpha \alpha'} = \mathbf{I}_{\alpha \alpha'}.
\end{align*}

4. Since $(e_\alpha)_{\alpha \in V}$ is a Hilbert basis of $\mathscr{E}_h$, the operator $\Lh$ restricted to $\mathscr{E}_h$ is obviously unitary equivalent to $\M$, the unitary being
\[ \mathbf U : x \in \ell^2(V) \mapsto \sum_{\alpha} x_\alpha e_{\alpha} \in \mathscr{E}_h. \]
To obtain the estimates on $\M$, let us first introduce the infinite matrix $$\widetilde{\M}= ( \langle \Lh^V g_\alpha, g_\beta \rangle)_{\alpha, \beta \in V}.$$ We claim that $\widetilde{\M}$ defines a bounded operator on $\ell^1(V)$ such that
\begin{equation}\label{def.tildeL}
 \widetilde{\M} = \mu_1 \mathbf{I} + \mathbf W + \mu_1 \mathbf T + \mathscr{O}(\mathcal E_{\varepsilon,L}), 
\end{equation}
in $\ell^1$ operator norm topology. Indeed, for $\alpha = \beta$ we have
\begin{align*}
 \langle \Lh^V g_\alpha, g_\alpha \rangle &= \mu_1 \| g_\alpha \|^2 + \langle (\Lh^V - \mu_1) g_\alpha, g_\alpha \rangle\\
&=  \mu_1 \| g_\alpha \|^2 + \langle (\Lh^V - \mu_1) \Psi_\alpha, \Psi_\alpha \rangle +  \langle (\Lh^V - \mu_1) ( g_\alpha - \Psi_\alpha), ( g_\alpha- \Psi_\alpha) \rangle\\
&= \mu_1 + \mathscr{O}(e^{- \frac{2d(L-R-4\eta)}{h}}),
\end{align*}
by Lemmas \ref{lem.Psiv} and \ref{lem.gv}. For $L \leq |\alpha - \beta| \leq L(1 + \varepsilon)$ we have
\begin{align*}
 \langle \Lh^V g_\alpha, g_\beta \rangle \mathbf{1}_{L \leq |\alpha - \beta| \leq L(1 + \varepsilon)} &=   \Big( \langle (\Lh^V - \mu_1) g_\alpha, g_\beta \rangle  + \mu_1 \langle g_\alpha, g_\beta \rangle \Big) \mathbf{1}_{L \leq |\alpha - \beta| \leq L(1 + \varepsilon)}\\ &=\mathbf W_{\alpha \beta} + \mu_1 \mathbf{T}_{\alpha \beta} +  \mathscr{O}(e^{- \frac{2d(L-R-4\eta)}{h}}), 
\end{align*}
where we used Lemma \ref{lem.gv} again. Finally the last part we estimate as before,
\begin{multline*}
 \sup_{\alpha \in V} \sum_{\substack{\beta \in V \\ | \alpha - \beta| >L(1 + \varepsilon)}} |\langle \Lh^V g_\alpha, g_\beta \rangle | \\ \leq  \sup_{\alpha \in V} \sum_{\substack{\beta \in V \\ | \alpha - \beta| >L(1 + \varepsilon)}} |\langle \Lh^V (g_\alpha - \Psi_\alpha), g_\beta - \Psi_\beta \rangle | + |\langle \Lh^V \Psi_\alpha, \Psi_\beta \rangle |.
 \end{multline*}
The second term vanishes for all but finitely many $\beta$'s. For the remaining ones we use Lemma \ref{lem.wh2}. In the first term we insert the exponential weight and find
 \begin{align*}
 &\sup_{\alpha \in V} \sum_{\substack{\beta \in V \\ | \alpha - \beta| >L(1 + \varepsilon)}} |\langle \Lh^V g_\alpha, g_\beta \rangle | \\ &\qquad \leq  \sup_{\alpha \in V} \sum_{\substack{\beta \in V \\ | \alpha - \beta| >L(1 + \varepsilon)}} e^{-h \langle \alpha - \beta \rangle} \| e^{h \langle x - \alpha \rangle} \Lh^V (g_\alpha - \Psi_\alpha) \| \| e^{h \langle x - \beta \rangle} ( g_\beta - \Psi_\beta ) \| \\ & \qquad \qquad + \mathscr{O}(e^{- \frac{S_h(L(1+\varepsilon))}{h}}+ e^{-\frac{d(L-R- 3 \eta) + d(\varepsilon L + R + 2 \eta)}{h}}+e^{-\frac{2d(L-R-4\eta)}{h}}).
 \end{align*}
 We then use Lemma \ref{lem.weighted} to obtain
  \begin{align*}
 \sup_{\alpha \in V} \sum_{\substack{\beta \in V \\ | \alpha - \beta| >L(1 + \varepsilon)}} |\langle \Lh^V g_\alpha, g_\beta \rangle | \leq  \mathscr{O}(e^{- \frac{S_h(L(1+\varepsilon))}{h}}+ e^{-\frac{d(L-R- 3 \eta) + d(\varepsilon L + R + 2 \eta)}{h}}+e^{-\frac{2d(L-R-4\eta)}{h}}).
 \end{align*}
and this proves \eqref{def.tildeL}.

By definition of $\M$ and $e_{\alpha}$ we have
\begin{align*}
\M_{\alpha \alpha'} = \sum_{\beta \beta' \in V} \mathbf G^{- \frac 12}_{\alpha \beta} \mathbf G^{- \frac 12}_{\alpha' \beta'} \widetilde{\M}_{\beta \beta'} = (\G^{- \frac 12} \widetilde{\M} \G^{- \frac 12} )_{\alpha \alpha'}.
\end{align*}
Moreover using item (2) and \eqref{def.tildeL} we have
\[\G^{- \frac 12} \widetilde{\M} \G^{- \frac 12} = \mu_1 \mathbf I + \mathbf W + \mathscr{O}(\mathcal E_{\varepsilon,L}) ,\]
in $\ell^1$ operator norm topology. We deduce that the remainder $\mathbf S = \M - \mu_1 \mathbf I - \mathbf W$ is small enough in $\ell^1$ operator norm topology, but also in $\ell^{\infty}$ norm topology since it is symmetric. Using interpolation we deduce
\[  \M - \mu_1 \mathbf I - \mathbf W = \mathscr{O}(\mathcal E_{\varepsilon,L}) \]
in $\ell^2$ operator norm topology.
\end{proof}

\section{Reduction \texorpdfstring{when $e_0 = \frac 12$}{(Case 2)}} \label{sec.reduction2}

We now consider the case $e_0= \frac 12$, and we assume in all of this section that $h$ is as small as we want along a $\frac 12$-sequence. In this case, the first eigenvalue of the single-obstacle operator $\Lh^0$ might be double. However, Theorem \ref{thm.single} ensures that there is a gap with the third eigenvalue,
\begin{equation}
\mu_3 - \mu_2 \geq 2c_0 h^2 + o(h^2).
\end{equation} 
By Lemma \ref{lem.m-m+}, the first two eigenvalues are $\lbrace \mu_-,\mu_+ \rbrace$ where $\mu_\pm$ is the ground state energy of $\mathscr{L}_{h,m_\pm}$ and
\begin{align}
m_- &= \xi_h - \frac 12 + o(1) = \frac{BR^2}{2h} + \sqrt{\frac{\Theta_0 BR^2}{h}} + \mathcal C_2 - \frac 12 + o(1),\\
m_+ &= \xi_h + \frac 12 + o(1)= \frac{BR^2}{2h} + \sqrt{\frac{\Theta_0 BR^2}{h}} + \mathcal C_2 + \frac 12 + o(1).
\end{align}
Depending on $h$, either $\mu_-$ or $\mu_+$ can be the smallest. The associated eigenfunctions are $\Phi^-$ and $\Phi^+$.
We denote by $\Pi$ the spectral projector of $\Lh^V$ corresponding to the interval $(-\infty, \mu_2 +c_0 h^2)$. 

The reduction in this case follows essentially the same lines as in Section \ref{sec.reduction1}, except that the eigenspace of the single-obstacle operator is now bidimensional. Using Agmon estimates (Lemma \ref{lem.Agmon}), we know that the spectrum of $\Lh^V$ is exponentially close to $\mu_1$.

\begin{lemma}\label{lem.spectralgap2}
If $\lambda \in {\rm{sp}}(\Lh^V)$ is such that $\lambda < \mu_2 +c_0 h^2$ then we have 
\[ \min \big( |\lambda - \mu_-| , |\lambda - \mu_+| \big) \leq C_N h^N \]
 for all $N \geq 2$, where $C_N>0$ is independent of $\lambda$ and $h$. 
\end{lemma}

\begin{proof}
The proof is identical to Lemma \ref{lem.spectralgap}, up to changing the distance $|\lambda - \mu_1|$ by the minimal distance $\min \big( |\lambda - \mu_-| , |\lambda - \mu_+| \big)$.
\end{proof}

Recalling the definition \eqref{def.cutoff} of the cutoff function $\chi_\eta$, we define for $\alpha \in V$ the pairs of approximate eigenfunctions
\begin{equation}\label{def.Psialpha2}
\Psi_{\alpha}^{\pm} = \tau_\alpha^B \big( \chi_\eta \Phi^\pm \big), \quad g_{\alpha}^\pm = \Pi \Psi_\alpha^{\pm}.
\end{equation}

Again, the functions $\Psi_\alpha^{\pm}$ are good approximate solutions to the eigenvalue equation for $\Lh^V$. Moreover, $g_\alpha^\pm$ is exponentially close to $\Psi_\alpha^\pm$.

\begin{lemma}\label{lem.gv2}
If $\eta$ is small enough then
\begin{enumerate}
\item For all $\alpha \in V$,
\begin{align*}
 \|( \Lh^V - \mu_\pm) \Psi_\alpha^\pm \| = \mathscr{O}\big( e^{- \frac{d(L-R-3\eta)}{h}} \big), \quad \| \Psi_\alpha^\pm \| = 1 + \mathscr{O}\big( e^{- \frac{d(L-R-3\eta)}{h}} \big),
 \end{align*}
 and
 \[\langle \Psi_\alpha^+, \Psi_\alpha^{-} \rangle = \mathscr{O}\big( e^{- \frac{2d(L-R-3\eta)}{h}} \big). \]
\item For all $\alpha \in V$,
\[ \| g_\alpha^\pm - \Psi_\alpha^\pm \| = \mathscr{O}\big( e^{- \frac{d(L-R-4\eta)}{h}} \big), \quad \langle \Lh^V ( g_\alpha^\pm - \Psi_\alpha^\pm), g_\alpha^\pm - \Psi_\alpha^\pm \rangle = \mathscr{O}\big( e^{- \frac{d(L-R-4\eta)}{h}} \big),\]
and the estimates are uniform with respect to $\alpha \in V$.
\item The vectors $(g_\alpha^\pm)_{\alpha \in V}$ span a dense subspace of $\mathscr{E}_h$.
\end{enumerate}
\end{lemma}

\begin{proof}
The proof of the first item is identical to Lemma \ref{lem.Psiv}, except for the last statement. There, we simply use that $\langle \Phi^+, \Phi^- \rangle = 0$ and the decay of $\Phi^{\pm}$ to find
\[\langle \Psi_\alpha^+, \Psi_\alpha^- \rangle = \langle (\chi_\eta^2-1) \Phi^+, \Phi^- \rangle = \mathscr{O}\big( e^{- \frac{2d(L-R-3\eta)}{h}} \big).\]
The second item is identical to Lemma \ref{lem.gv}. To obtain that $(g_\alpha^\pm)_{\alpha \in V}$ span a dense subspace of $\mathscr E_h$, we also follow the same strategy. We can show that any $\Psi \in \lbrace g_\alpha^\pm\,, \alpha \in V \rbrace^\perp$ satisfy
\[ \langle \Lh^V \Psi, \Psi \rangle \geq \Big( \mu_2 + 2c_0 h^2 -C h^N\Big) \| \Psi \|^2,\]
using the spectral projector
\[\widetilde{\Pi}_0 = \langle \cdot, \Phi^+ \rangle \Phi^+ + \langle \cdot , \Phi^- \rangle \Phi^-\]
instead of $\Pi_0$, and the lower bound $\mu_3 - \mu_2 \geq 2 c_0 h^2$. If $\Psi \in \mathscr{E}_h$ we also have $\langle \Lh^V \Psi, \Psi \rangle \leq (\mu_2 + c_0 h^2) \| \Psi \|^2$, and
therefore $\Psi$ must have vanishing norm -- since $\mu_1 = \mu_2 + o(h^2)$.
\end{proof}

The proof of Lemma \ref{lem.weighted} also works here to provide some decay on $g_\alpha^\pm - \Psi_\alpha^\pm$,
\begin{equation}\label{eq.weighted2}
\| e^{ h \langle x - \alpha \rangle} ( g_\alpha^\pm - \Psi_\alpha^\pm) \| = \mathscr{O}( e^{-\frac{d(L-R-4\eta)}{h}} ),  \quad \| e^{ h \langle x - \alpha \rangle} \Lh^V ( g_\alpha^\pm - \Psi_\alpha^\pm) \| = \mathscr{O}( e^{-\frac{d(L-R-4\eta)}{h}} ).
\end{equation}

We are now in position to prove the main result of this section, which is the analog of Theorem \ref{thm.Gram}. We define effective operators which act on the space $\ell^p(V) \otimes \C^2$. An element of this space will be denoted by $x= (x^+, x^-)$, where $x^+$, $x^- \in \ell^p(V)$. Operators acting on this space can be described by an infinite matrix $\mathbf M = (\mathbf M_{\alpha \beta}^{\sigma \sigma'})_{\alpha \beta \in V, \sigma \sigma' \in \lbrace \pm \rbrace}$, and $y = \mathbf M x$ is then such that
\[y^\sigma_\alpha = \sum_{\beta \in V, \sigma' \in \lbrace \pm \rbrace} \mathbf M_{\alpha \beta}^{\sigma \sigma'} x^{\sigma'}_\beta.\]
We recall the size of the error terms,
\[\mathcal E_{\varepsilon,L} = e^{- \frac{2S_h(L)}{h}} + e^{-\frac{2d(L-R-4\eta)}{h}}+ e^{-\frac{S_h(L(1+\varepsilon))}{h}} + e^{-\frac{d(L-R-4\eta) + d(\varepsilon L + R + 2 \eta)}{h}}, \]
which are controlled in Lemma \ref{lem.errors}.

\begin{theorem}\label{thm.Gram2}
Let $\varepsilon >0$ and assume $e_0=\frac 12$. For all $p \in [1,\infty]$, the infinite matrix $\mathbf G = (\langle g_\alpha^\sigma, g_\beta^{\sigma'} \rangle)_{\alpha \beta \in V, \sigma \sigma' \in \lbrace \pm \rbrace }$ defines a bounded operator on $\ell^p(V) \otimes \C^2$. Moreover,
\begin{enumerate}
\item We have the following estimate in $\ell^p$ operator norm topology
\[\G = \mathbf I + \T + \mathscr{O}(\mathcal E_{\varepsilon,L}), \quad {\text{in}}\quad \mathcal{L}(\ell^p(V)\otimes \C^2),\]
where $\T_{\alpha \beta}^{\sigma \sigma'} = \langle g_\alpha^\sigma, g_\beta^{\sigma'} \rangle \one_{L \leq |\alpha - \beta| \leq L(1 + \varepsilon )}$. 
\item $\G^{-1}$ is a bounded operator on $\ell^p(V)$, and has a square root satisfying
\[ \G^{- \frac 1 2} = \mathbf I - \frac 1 2 \T + \mathscr{O}(\mathcal E_{\varepsilon,L}), \quad {\text{in}}\quad \mathcal{L}(\ell^p(V) \otimes \C^2). \]
\item The vectors
\[e_\alpha^\sigma = \sum_{\beta \in V, \sigma' \in \lbrace \pm \rbrace} (\G^{- \frac 12})_{\alpha \beta}^{\sigma \sigma'} g_\beta^{\sigma'}, \qquad \alpha \in V, \quad \sigma \in \lbrace \pm \rbrace,\]
define a Hilbert basis of $\mathscr{E}_h$.
\item The restriction of $\Lh^V$ to $\mathscr{E}_h$ is unitarily equivalent to the operator\\ $\M = ( \langle \Lh^V e_\alpha^\sigma, e_\beta^{\sigma'} \rangle)_{\alpha, \beta \in V, \sigma, \sigma' \in \lbrace \pm \rbrace}$ acting on $\ell^2(V) \otimes \C^2$, and
\[ \M = \mathbf D + \mathbf W +\mathscr{O}(\mathcal E_{\varepsilon,L}),\]
where $\mathbf D_{\alpha \beta}^{\sigma \sigma'} = \mu_\sigma \mathbf I_{\alpha \beta}^{\sigma \sigma'}$ and
\begin{equation}\label{eq.Wabs}
 \mathbf W_{\alpha \beta}^{\sigma \sigma'} = \big\langle (\Lh - \frac{\mu_\sigma + \mu_{\sigma'}}{2}) \Psi_\alpha^\sigma, \Psi_\beta^{\sigma'} \big\rangle \one_{L \leq |\alpha - \beta | \leq L(1 + \varepsilon)}.
 \end{equation}
\end{enumerate}
\end{theorem}

\begin{proof}
1. As before, we can use the $\ell^1$ operator norm. We decompose $\mathbf G$ as
\[ \mathbf G = \widetilde{\mathbf I} + {\mathbf T} + \widetilde{\mathbf R},\]
where
\begin{align*}
\widetilde{\mathbf I}_{\alpha \beta}^{\sigma \sigma'} = \langle g_\alpha^\sigma, g_\beta^{\sigma'} \rangle \mathbf{1}_{\alpha = \beta}, \quad \mathbf T_{\alpha \beta}^{\sigma \sigma'} = \langle g_\alpha^\sigma, g_\beta^{\sigma'} \rangle \mathbf{1}_{L \leq |\alpha - \beta| \leq L(1+\varepsilon)}, \quad \widetilde{\mathbf R}_{\alpha \beta}^{\sigma \sigma'} = \langle g_\alpha^\sigma, g_\beta^{\sigma'} \rangle \mathbf{1}_{|\alpha - \beta| > L(1+\varepsilon)}.
\end{align*}
As in \eqref{eq.bdR}, we bound $\widetilde{\mathbf R}$ and $\mathbf I - \widetilde{\mathbf I}$ using Lemma \ref{lem.gv2}, equation \eqref{eq.weighted2} and Lemma \ref{lem.scalarproduct}. Note that to estimate the off-diagonal terms of $\widetilde{\mathbf I}$ we need the estimate on $\langle \Psi_\alpha^+, \Psi_\alpha^- \rangle$ from Lemma \ref{lem.gv2}.

2. The proof of the second item is identical to Theorem \ref{thm.Gram}. Indeed, $\mathbf G$ is still a perturbation of the identity, and a Taylor expansion of the square root gives the result.

3. The third item is also unchanged. $e_\alpha^\sigma$ is a perturbation of $g_\alpha^\sigma$, and these vectors span a dense subspace of $\mathscr E_h$. The family $(e_\alpha^\sigma)_{\alpha \in V, \sigma \in \lbrace \pm \rbrace}$ is also orthonormal since
\begin{align*}
\langle e_\alpha^\sigma, e_{\alpha'}^{\sigma'} \rangle &= \sum_{\beta, \beta' \in V \tau, \tau' \in \lbrace \pm \rbrace } (\G^{- \frac 1 2})_{\alpha \beta}^{\sigma \tau}  \langle g_\beta^\tau, g_{\beta'}^{\tau'} \rangle (\G^{- \frac 1 2})_{\alpha' \beta'}^{\sigma' \tau'} = (\G^{- \frac 12} \G \G^{- \frac 12})_{\alpha \alpha'}^{\sigma \sigma'} = \mathbf{I}_{\alpha \alpha'}^{\sigma \sigma'},
\end{align*}
and therefore it is a Hilbert basis of $\mathscr E_h$.

4. $\Lh^V$ is unitarily equivalent to $\M$ through the Hilbert basis $(e_\alpha^\sigma)$. Before estimating $\M$, we introduce the operator $\widetilde{\M} = ( \langle \Lh^V g_\alpha^\sigma, g_\beta^{\sigma'} \rangle)_{\alpha \beta \in V, \sigma \sigma' \in \lbrace \pm \rbrace}$. We claim that
\begin{equation}\label{eq.Ltilde2}
\widetilde{\M} = \mathbf D + \mathbf W + \frac{1}{2}( \mathbf D \mathbf T + \mathbf T \mathbf D) + \mathscr{O} \big(\mathcal E_{\varepsilon, L} \big),
\end{equation}
in $\ell^1(V) \otimes \C^2$ operator norm topology, with $\mathbf D_{\alpha \beta}^{\sigma \sigma'} = \mu_\sigma \mathbf{I}_{\alpha \beta}^{ \sigma \sigma'}$. Indeed, on the diagonal $\alpha = \beta$ we have
\begin{align*}
 \langle \Lh^V g_\alpha^\sigma, g_\alpha'^{\sigma'} \rangle &= \mu_\sigma \langle g_\alpha^\sigma, g_\alpha^{\sigma'} \rangle + \mathscr{O}(e^{-\frac{2d(L-R-4\eta)}{h}}) \\
 &= \mu_\sigma \mathbf I_{\alpha \alpha}^{\sigma \sigma'} + \mathscr{O}(e^{-\frac{2d(L-R-4\eta)}{h}})
 \end{align*}
 because of Lemma \ref{lem.gv2}. For $L \leq |\alpha - \beta | \leq L(1+\varepsilon)$, we have
 \begin{align*}
 \langle \Lh^V g_\alpha^\sigma, g_\beta^{\sigma'} \rangle &= \langle \Lh \Psi_\alpha^\sigma, \Psi_\beta^{\sigma'} \rangle + \mathscr{O}(e^{-\frac{2d(L-R-4\eta)}{h}}) \\
 &=  \mathbf W_{\alpha \beta}^{\sigma \sigma'} + \frac{\mu_\sigma + \mu_{\sigma'}}{2} \langle \Psi_\alpha^\sigma, \Psi_\beta^{\sigma'} \rangle+ \mathscr{O}(e^{-\frac{2d(L-R-4\eta)}{h}})\\
 &= \mathbf W_{\alpha \beta}^{\sigma \sigma'} + \frac{\mu_\sigma + \mu_{\sigma'}}{2} \mathbf T_{\alpha \beta}^{\sigma \sigma'} +\mathscr{O}(e^{-\frac{2d(L-R-4\eta)}{h}}),
 \end{align*}
 where we used Lemma \ref{lem.gv2} again in the first line. Finally, the remaining terms $|\alpha - \beta |> L(1+\varepsilon)$ are bounded exactly as in Theorem \ref{thm.Gram}, using equation \eqref{eq.weighted2} and Lemma \ref{lem.gv}. This proves \eqref{eq.Ltilde2}. Now, by definition of $\M$ and $e_\alpha^\sigma$ we have
 \begin{align*}
\M_{\alpha \alpha'}^{\sigma \sigma'} = \sum_{\beta \beta' \in V, \tau, \tau' \in \lbrace \pm \rbrace} (\mathbf G^{- \frac 12})_{\alpha \beta}^{\sigma \tau}( \mathbf G^{- \frac 12})_{\alpha' \beta'}^{\sigma' \tau'} \widetilde{\M}_{\beta \beta'}^{\tau \tau'} = (\G^{- \frac 12} \widetilde{\M} \G^{- \frac 12} )_{\alpha \alpha'}^{\sigma \sigma'}.
\end{align*}
Moreover using item (2) and \eqref{eq.Ltilde2} we deduce
\[\M = \G^{- \frac 12} \widetilde{\M} \G^{- \frac 12} = \mathbf D + \mathbf W + \mathscr{O}(\mathcal E_{\varepsilon,L}) ,\]
in $\ell^1$ (and thus $\ell^2$) operator norm topology.
\end{proof}

\section{Proof of Theorems \ref{thm.two.obstacles} and \ref{thm.two.obstacles.2}}\label{sec:tunneling_2obstacles}

In this section we prove the spectral gap estimates in the case of two obstacles. The centers of the obstacles are $V= \lbrace \alpha_\ell, \alpha_r \rbrace$ with
\[ \alpha_\ell = \Big( - \frac L2,0 \Big), \qquad \alpha_r = \Big( \frac L2, 0 \Big).\]

\begin{proof}[Proof of Theorem~\ref{thm.two.obstacles}]
When $e_0 \neq \frac 12$, we use Theorem \ref{thm.Gram} with $\varepsilon \geq 2$. In this case, $\mathscr{E}_h$ is two-dimensional and the restriction of $\Lh^V$ to $\mathscr{E}_h$ is unitarily equivalent to the matrix
\[ \M = 
\begin{pmatrix}
\mu_1 & \mathbf W_{\alpha_\ell \alpha_r} \\
\mathbf W_{\alpha_r \alpha_\ell}& \mu_1 
\end{pmatrix}
+ \mathscr{O}\big( \mathcal{E}_{\varepsilon,L} \big).
 \]
The first two eigenvalues of $\Lh^V$ are thus of the form
\[ \lambda_1 = \mu_1 - \big|\mathbf W_{\alpha_\ell \alpha_r} \big| + \mathscr{O}\big( \mathcal{E}_{\varepsilon,L} \big), \qquad  \lambda_2 = \mu_1 + \big|\mathbf W_{\alpha_\ell \alpha_r} \big| + \mathscr{O}\big( \mathcal{E}_{\varepsilon,L} \big). \]
In Lemma \ref{lem.wh1} and \ref{lem.wh2} below the interaction coefficient $\mathbf W_{\alpha_\ell \alpha_r}$ is estimated and we find
\begin{align*}
 \lambda_2 - \lambda_1 &= 4 |{\rm{Re}}(w_L)| + \mathscr{O}\big( \mathcal{E}_{\varepsilon,L} \big) \\
 &= 4 |{\rm{Re}}(C_L)| K_L^{-e_0-1} h e^{-\frac{S_h(L)}{h}} \big( 1+o(1) \big)+ \mathscr{O}\big( \mathcal{E}_{\varepsilon,L} \big).
\end{align*}
The error terms $\mathcal{E}_{\varepsilon,L}$ are shown to be small enough in Lemma \ref{lem.errors}, using that $L>6R$.
\end{proof}

\begin{proof}[Proof of Theorem~\ref{thm.two.obstacles.2}]
When $e_0 = \frac 12$, the eigenspace $\mathscr{E}_h$ is four-dimensional and we use Theorem \ref{thm.Gram2} instead. The Hilbert space $\ell^2(V) \otimes \mathbf C^2$ is four-dimensional as well, and $\Lh^V$ is unitarily equivalent to the matrix
\[ \M = 
\begin{pmatrix}
\mu_-& \mathbf W^{--}& 0 & \mathbf{W}^{-+} \\
\mathbf W^{--} & \mu_- & \overline{\mathbf{W}}^{+-} & 0 \\
0 & \mathbf W^{+-} & \mu_+ & \mathbf{W}^{++} \\
\overline{\mathbf W}^{-+} & 0 & \mathbf{W}^{++} & \mu_+
\end{pmatrix}
 + \mathscr{O}\big( \mathcal{E}_{\varepsilon,L} \big),
 \]
 where we used the short-hand notation $\mathbf W^{\sigma \sigma'} = \mathbf{W}^{\sigma \sigma'}_{\alpha_\ell \alpha_r}$. We use the Lemmas \ref{lem.wh1} and 
\ref{lem.wh2} to estimate the interaction coefficients and find
\[ \M =  
\begin{pmatrix}
\mu_- & \lambda_h K_L^{-1} & 0 & -\lambda_h \\
\lambda_h K_L^{-1} & \mu_- & \lambda_h & 0 \\
0 &\lambda_h& \mu_+ &- \lambda_h K_L \\
-\lambda_h & 0 & -\lambda_h K_L & \mu_+
\end{pmatrix}
 + o \big( h e^{- \frac{S_h(L)}{h}} \big),
 \]
 with $\lambda_h =2 (-1)^{m_+} {\rm{Re}}(C_L) K_L^{-e_0}$.
In the special case when $\mu_- = \mu_+$ we find
\begin{align*}
\lambda_1 &= \mu_1 - |\lambda_h| \big( K_L + K_L^{-1} \big) + o \big( h e^{- \frac{S_h(L)}{h}} \big),\\
\lambda_2 &= \mu_1+ o \big( h e^{- \frac{S_h(L)}{h}} \big),\\
\lambda_3 &= \mu_1+o \big( h e^{- \frac{S_h(L)}{h}} \big) ,\\
\lambda_4 &= \mu_1 + |\lambda_h| \big( K_L + K_L^{-1} \big) +o \big( h e^{- \frac{S_h(L)}{h}} \big),
\end{align*}
and this concludes the proof of Theorem \ref{thm.two.obstacles}.
\end{proof}

\section{tunneling for a lattice of obstacles}\label{sec:tunneling_lattice}
In this section be calculate the tunneling in the case of infinitely many disjoint obstacles.
We focus here on the special case when the obstacles are distributed along a lattice, that is $V= L \Z^2$.
\begin{proof}[Proof of Theorem~\ref{thm.lattice}]
Let us start with the simpler situation, when $e_0 \neq \frac 12$. We use Theorem~\ref{thm.Gram}, and make the specific choice $\varepsilon = 0.4$. Indeed, we want $\varepsilon < \sqrt{2}-1$ to ensure that we do not capture the interaction between obstacles along diagonals. For instance, the lattice points $(0,0)$ and $(L,L)$ have distance $L\sqrt{2} > L(1+ \varepsilon)$ to each other: their interaction is not included in the infinite matrix $\mathbf W$. We also want $\varepsilon$ as large as possible in order to reduce the size of the error terms $\mathcal E_{\varepsilon,L}$. Note that we could also choose $\varepsilon$ larger and take into account the diagonal interactions, but this would make the interaction matrix $\mathbf W$ more complicated.

With the choice $\varepsilon = 0.4$, we obtain that the restriction of $\Lh^V$ to $\mathscr{E}_h$ is unitarily equivalent to an operator $\M$ on $\ell^2(L\Z^2)$ of the form
\[ \M = \mu_1 \mathbf I + \mathbf W + \mathscr{O}(\mathcal E_{\varepsilon,L}).\]
For $\alpha$, $\beta \in L \Z^2$, the coefficients $\mathbf W_{\alpha \beta}$ are all vanishing unless $|\alpha - \beta|=L$ (i.e. they are closest neighbours). In this case, by Lemmas \ref{lem.wh1} and \ref{lem.wh2} we have
\[ \mathbf W_{\alpha \beta} = (-1)^{m_-+1} e^{\frac{iB}{2h}\alpha \wedge \beta} 2C_L  K_L^{-e_0-1} h e^{- \frac{S_h(L)}{h}} \big( 1+ o(1) \big).\]
 With the notations of Theorem \ref{thm.lattice} this gives
\[ \mathbf W = \lambda \X + o \big( h e^{-\frac{S_h(L)}{h}} \big) ,\]
where 
\[ \X_{\alpha \beta} = \begin{cases}
e^{\frac{iB}{2h}\alpha \wedge \beta} &{\rm{if}} \, |\alpha - \beta |=L,\\
0 &{\rm{otherwise}}.
\end{cases}\]
The errors $\mathcal{E}_{\varepsilon,L}$ are shown to be small enough in Lemma \ref{lem.errors}, as soon as $L > 25 R$.

\medskip
We now consider $e_0 = \frac 12$. In this case we use Theorem \ref{thm.Gram2} (still with $\varepsilon = 0.4$) and we obtain an operator on $\ell^2(L\Z^2) \otimes \mathbf C^2$ with similar expansion,
\[ \M = \mathbf D + \mathbf W + \mathscr{O}(\mathcal E_{\varepsilon,L}).\]
The coefficients $\mathbf W_{\alpha \beta}^{\sigma \sigma'}$ are vanishing unless $|\alpha - \beta|=L$, in which case we have
\[ \mathbf W_{\alpha\beta}^{\sigma \sigma'} = (-1)^{m_{\sigma'}+1} e^{\frac{iB}{2h}\alpha \wedge \beta} 
e^{- i \frac{\sigma - \sigma'}{2} \arg(\beta-\alpha)}  2C_L K_L^{\frac{\sigma + \sigma' -1}{2}} h e^{-\frac{S_h(L)}{h}} \big( 1+o(1) \big),\]
by Lemmas \ref{lem.wh1} and \ref{lem.wh2}. With the notations of Theorem \ref{thm.lattice} we have $\mathbf W = \lambda \Y+ o \big( h e^{-\frac{S_h(L)}{h}} \big)$. Finally, the errors $\mathcal E_{\varepsilon,L}$ are small enough by Lemma \ref{lem.errors}. 
\end{proof}

\appendix

\section{Analysis of the phase} \label{sec.phase}

The phase function involved in several integrals in this Appendix is
\begin{equation}
\begin{aligned}
s_h^{\sigma \sigma'}(x_1,x_2)&= \frac{B}{4} \Big(x_1+\frac{\ell}{2} \Big)^2 + \frac{B}{4}\Big(x_1- \frac{\ell}{2}\Big)^2 + \frac{B}{2} x_2^2 - \frac{BR^2}{2} + i \frac{B\ell x_2}{2} \\ &\qquad - hm_{\sigma} \log \Big( \frac{\ell/2 +x_1+ix_2}{R} \Big) - hm_{\sigma'} \log \Big(\frac{ \ell/2 - x_1 +ix_2}{R} \Big),
\end{aligned}
\end{equation}
where $\sigma$, $\sigma' \in \lbrace \pm \rbrace$. It depends on the parameter $\ell > 2R$. Recalling that $m_\pm$ depends on $h$, we consider the main part of this function,
\begin{equation}
\begin{aligned}
s_0(x_1,x_2)&= \frac{B}{4} \Big(x_1+\frac{\ell}{2} \Big)^2 + \frac{B}{4}\Big(x_1- \frac{\ell}{2}\Big)^2 + \frac{B}{2} x_2^2 - \frac{BR^2}{2} + i \frac{B\ell x_2}{2} \\ &\qquad -\frac{BR^2}{2} \log \Big( \frac{\ell/2 +x_1+ix_2}{R} \Big) - \frac{BR^2}{2} \log \Big(\frac{ \ell/2 - x_1 +ix_2}{R} \Big) .
\end{aligned}
\end{equation}
We collect below the main properties of these functions, necessary to use the Laplace method.

\begin{lemma}\label{lem.phase}
The function $(x_1,x_2) \mapsto s_0(x_1,x_2)$  has the following properties.
\begin{enumerate}
\item $s_0$ is holomorphic on $D = \lbrace {\rm{Re}}(x_1) \in \big( - \frac \ell 2, \frac \ell 2 \big)$, ${\rm{Im}}(x_2) \leq 0 \rbrace \subset \mathbf C^2$. 
\item $s_0$ has a unique critical point $(x_1^*,x_2^*)$ in $D$, which is
\[x_1^* = 0,  \quad x_2^* = -i  \sqrt{\frac{\ell^2}{4}-R^2}.\]
\item The critical value is
\[ s_0(x_1^*,x_2^*) = \frac{B \ell}{4}\sqrt{\ell^2-4R^2} -  BR^2 \ln \Big( \frac{\ell+ \sqrt{\ell^2-4R^2}}{2R}  \Big).\]
\item The Hessian is such that ${\rm{Im}} \nabla^2 s_0(x^*) = 0$. Also, $\partial_1 \partial_2 s_0(x^*) = 0$ and
\[ \begin{cases}
 \partial_1^2 s_0(x^*) &= \frac{B\ell}{2R^2}\big( \ell - \sqrt{\ell^2-4R^2} \big), \\ 
 \partial_2^2 s_0(x^*) &= \frac{B}{2R^2}\sqrt{\ell^2-4R^2} \big( \ell - \sqrt{\ell^2-4R^2} \big).
 \end{cases} \]
\item The function $x_1 \in (-\frac \ell 2, \frac \ell 2) \mapsto s_0(x_1,x_2^*)$ is real valued, and has a unique global minimum at $0$.
\item For $x_2 \in \mathbf R$, we have the following expansion as $x_2 \to 0$,
\begin{align*}
 s_h^{\sigma \sigma'}(0,x_2+x_2^*) &= s_h^{\sigma \sigma'}(0,x_2^*) + \frac{B \sqrt{\ell^2-4R^2}}{4R^2} ( \ell - \sqrt{\ell^2-4R^2} ) x_2^2 \\ & \qquad - i \frac{x_2}{R} (\ell- \sqrt{\ell^2-4R^2}) \sqrt{h \Theta_0 B} + \mathscr{O}(x_2^3 + \sqrt h x_2^2 + h x_2),
 \end{align*}
 and
 \begin{align*} s_h^{\sigma \sigma'}(0,x_2^*) &=  \frac{B\ell}{4} \sqrt{\ell^2-4R^2} - h (m_\sigma + m_{\sigma'}) \ln \Big( \frac{ \ell + \sqrt{\ell^2-4R^2} }{2R} \Big).  \\
 &= S_h(\ell) - h ( 2 \mathcal C_2 + m_\sigma + m_{\sigma'} - 2 \xi_h) \ln \Big( \frac{ \ell + \sqrt{\ell^2-4R^2} }{2R} \Big),
 \end{align*}
where we recall
\begin{equation}
S_h(\ell) =  \frac{B\ell}{4}\sqrt{\ell^2 - 4R^2} - \big( BR^2 + 2\sqrt{h\Theta_0 BR^2} \big) \ln \Big( \frac{\ell + \sqrt{\ell^2 - 4R^2}}{2R}\Big) .
\end{equation}
\end{enumerate}
\end{lemma}

\section{Estimates on the interaction coefficients}\label{appendix:IC}

We collect in this appendix estimates on the scalar products
\[ \langle \Psi_\alpha, \Psi_\beta \rangle, \qquad \langle \Psi_\alpha^\sigma, \Psi_\beta^{\sigma'} \rangle, \]
and on the interaction coefficients
\[ \mathbf W_{\alpha \beta} = \langle \big( \Lh - \mu_1 \big) \Psi_\alpha, \Psi_\beta \rangle, \qquad \mathbf W_{\alpha \beta}^{\sigma \sigma'} = \big\langle \big( \Lh - \frac{\mu_\sigma + \mu_{\sigma'}}{2} \big) \Psi_\alpha^\sigma, \Psi_\beta^{\sigma'} \big\rangle, \]
where $\alpha$, $\beta \in V$ and $\sigma$, $\sigma' \in \lbrace \pm \rbrace$. These estimates follow from the Laplace method, and the involved phase $s_h^{\sigma \sigma'}$ was described in Appendix \ref{sec.phase}. We start by estimating the scalar products in Lemma \ref{lem.scalarproduct}. We then explain how to rewrite the interaction coefficients by an interaction by part argument in Lemma \ref{lem.wh1}. We finally estimate these coefficients in Lemma \ref{lem.wh2}.

\begin{lemma}\label{lem.scalarproduct}
Let $\alpha$, $\beta \in V$ such that $|\alpha - \beta| = \ell \geq L$. Then our quasimodes $\Psi_\alpha$ and $\Psi_\alpha^\pm$ are such that:
\begin{enumerate}
\item When $e_0 \neq \frac 12$ we have
\[ \langle \Psi _\alpha, \Psi_{\beta} \rangle = \mathscr{O} \big(h^{\frac 12} e^{-\frac{S_h(\ell)}{h}} + e^{-\frac{d(L-R-3\eta) + d(\ell - L + R + 2 \eta)}{h}} \big),\]
as $h \to 0$ along a $e_0$-sequence, and uniformly with respect to $\alpha$ and $\beta$.
\item When $e_0 = \frac 12$, for $\sigma$, $\sigma' \in \lbrace \pm \rbrace$ we have
\[ \langle \Psi _\alpha^\sigma, \Psi_{\beta}^{\sigma'} \rangle = \mathscr{O} \big(h^{\frac 12} e^{-\frac{S_h(\ell)}{h}} + e^{-\frac{d(L-R-3\eta) + d(\ell - L + R + 2 \eta)}{h}} \big),\]
as $h \to 0$ along a $\frac 12$-sequence, and uniformly with respect to $\alpha$ and $\beta$.
\end{enumerate}
\end{lemma}

\begin{proof}
Note that when $e_0 \neq \frac 12$, $\Psi_\alpha$ is equal to $\Psi_\alpha^{-}$. Therefore, we can focus on the second case. We first use the definition \eqref{def.Psialpha2} of $\Psi_\alpha^\sigma$ and properties of the magnetic translations, in order to center the integral at $0$. Indeed,
\begin{align*}
\langle \Psi_\alpha^\sigma, \Psi_\beta^{\sigma'} \rangle &= \langle \tau_\alpha^B \chi_\eta \Phi^\sigma, \tau_\beta^B \chi_\eta \Phi^{\sigma'}\rangle = \langle \tau^B_{-\frac{\alpha + \beta}{2}} \tau^B_\alpha \chi_\eta  \Phi^\sigma,  \tau^B_{-\frac{\alpha + \beta}{2}} \tau^B_\beta \chi_\eta \Phi^{\sigma'} \rangle \\
&=e^{\frac{iB}{2h} \alpha \wedge \beta} \langle \tau_{\frac{\alpha - \beta}{2}}^B \chi_\eta \Phi^\sigma, \tau^B_{\frac{\beta - \alpha}{2}} \chi_\eta \Phi^{\sigma'} \rangle.
\end{align*}
Now we use a rotation $R_{-\theta}$ such that $R_{-\theta} \big( \frac{\alpha - \beta}{2} \big) = (- \frac \ell 2, 0)$. The commutation with magnetic translations gives
\[ R_\theta \cdot \tau_{\alpha}^B = \tau_{(R_{-\theta} \alpha)}^B \cdot R_\theta,\]
and therefore 
\begin{align*}
\langle \Psi_\alpha^\sigma, \Psi_\beta^{\sigma'} \rangle = e^{\frac{iB}{2h} \alpha \wedge \beta} \langle \tau^B_{(- \frac{\ell}{2},0)} R_\theta \chi_\eta \Phi^\sigma, \tau^B_{( \frac{\ell}{2},0)} R_\theta \chi_\eta \Phi^{\sigma'}\rangle.
\end{align*}
Then, using the angular dependence of $\Phi^\sigma$ we have $R_\theta \Phi^\sigma(x) =\Phi^\sigma( R_\theta x) = e^{im_\sigma \theta} \Phi^\sigma(x)$ and 
\begin{equation}
\langle \Psi_\alpha^\sigma, \Psi_\beta^{\sigma'} \rangle = e^{\frac{iB}{2h} \alpha \wedge \beta + i (m_\sigma - m_{\sigma'})\theta} \langle \tau^B_{(-\frac{\ell}{2},0)} \chi_\eta \Phi^\sigma,  \tau^B_{(\frac{\ell}{2},0)} \chi_\eta \Phi^{\sigma'} \rangle.
\end{equation}
We are reduced to the case of two disks on the horizontal axis, with relative distance $\ell$. For simplicity we write $\epsilon = e^{i(m_\sigma - m_{\sigma'})\theta}$, which is either equal to $1$ or $e^{\pm i\theta}$. Now we can use the explicit formula for $\Phi^\sigma$ from Lemma \ref{lem.decay},
\begin{equation}
\Phi^\sigma(x) = \Big( \frac{BR^2}{h} \Big)^{\frac{\Theta_0}{4}} \Big( \frac z R \Big)^{m_\sigma} e^{- \frac{B}{4h}(|z|^2 - R^2)} u_h(|z|), 
\end{equation}
where we identify $x=(x_1,x_2)$ and $z=x_1+ix_2$. Replacing we find
\begin{multline*}
 \langle \Psi_\alpha^\sigma, \Psi_\beta^{\sigma'} \rangle = e^{\frac{iB}{2h} \alpha \wedge \beta} e^{i \pi m_{\sigma'}}  \epsilon \Big( \frac{BR^2}{h} \Big)^{\frac{\Theta_0}{2}} \int_{\R^2} e^{- \frac{iB\ell x_2}{2h}} \chi_\eta \Big(z+ \frac{\ell}{2} \Big) \chi_\eta \Big(z-\frac{\ell}{2} \Big) \\ 
 \Big(\frac{\ell/2 +z}{R}\Big)^{m_\sigma} \overline{\Big(\frac{ \ell/2 - z}{R} \Big)^{m_{\sigma'}}} e^{- \frac{B}{4h} ( |z+\frac{\ell}{2}|^2 + |z- \frac{\ell}{2}|^2-2R^2)} u_h\Big(\big|z+\frac{\ell}{2}\big|\Big)u_h \Big(\big|z-\frac{\ell}{2}\big|\Big) \dd z. 
\end{multline*}
For simplicity we will not specify the arguments of $\chi_\eta$ in the notation below. Also, we can write the power of any complex number as exponential of a logarithm. The formula holds for any choice of complex logarithm, but on the support of $\chi_\eta$, the standard logarithm is continuous - and thus, we make this choice. We obtain
\begin{align*}
| \langle \Psi_\alpha^{\sigma}, \Psi_\beta^{\sigma'} \rangle | =\Big( \frac{BR^2}{h} \Big)^{\frac{\Theta_0}{2}} \Big| \int_{\R^2} e^{- \frac{s_h^{\sigma \sigma'}(z)}{h}} \chi_\eta u_h \Big(\big|z+\frac{\ell}{2}\big|\Big) \chi_\eta u_h \Big(\big|z-\frac{\ell}{2}\big|\Big) \dd z \Big|,
\end{align*}
with phase
\begin{align*}
s_h^{\sigma \sigma'}(x_1,x_2)&= \frac{B}{4} \big|z+\frac{\ell}{2} \big|^2 + \frac{B}{4}\big|z- \frac{\ell}{2}\big|^2 - \frac{BR^2}{2} + i \frac{B\ell x_2}{2} \\ &\qquad - hm_\sigma \log \Big( \frac{\ell/2 +z}{R} \Big) - hm_{\sigma'} \log \Big(\frac{ \ell/2 - \bar z}{R} \Big) .
\end{align*}
Note that ${\rm{Re}}(s_h^{\sigma \sigma'}(z)) \geq d \Big(\big|z+\frac{\ell}{2} \big| \Big) + d \Big(\big|z-\frac{\ell}{2} \big| \Big)$. We want to remove the cutoff function $\chi_\eta$. To do so, we use that on the support of $1- \chi_\eta$ we have \[{\rm{Re}}(s_h^{\sigma \sigma'}(z))  \geq d(L-R-2\eta) + d(\ell - L + R + 2 \eta).\] For the same reason, we can add a cutoff $\tilde{\chi}(x_1)$ supported on $|x_1|< \frac{\ell}{2} - R - \eta$ and we get
\begin{align*}
| \langle \Psi_\alpha^\sigma, \Psi_\beta^{\sigma'} \rangle | &= \Big( \frac{BR^2}{h} \Big)^{\frac{\Theta_0}{2}} \Big| \int_{\R^2} \tilde{\chi}(x_1) e^{- \frac{s_h^{\sigma \sigma'}(z)}{h}}  u_h \Big(\big|z+\frac{\ell}{2}\big|\Big)  u_h \Big(\big|z-\frac{\ell}{2}\big|\Big) \dd z \Big| \\ &\qquad + \mathscr{O}(e^{-\frac{d(L-R-3\eta)  + d(\ell - L + R + 2 \eta)}{h}}).
\end{align*}
This last cutoff is added to make sure the complex logarithms remain continuous. We can estimate this integral using the complex Laplace method. The phase $s_h^{\sigma \sigma'}(x_1,x_2)$ is analyzed in Lemma \ref{lem.phase}. It has a unique critical point of the form $(0,x_2^*)$ where $x_2^*$ is purely imaginary. The integral being analytic as function of $x_2$, we can make the complex shift $x_2 \mapsto x_2 + x_2^*$,
\begin{align*}
| \langle \Psi_\alpha^\sigma, \Psi_\beta^{\sigma'} \rangle | &=\Big( \frac{BR^2}{h} \Big)^{\frac{\Theta_0}{2}} \Big| \int_{\R^2} \tilde{\chi}(x_1) e^{- \frac{s_h^{\sigma \sigma'}(x_1,x_2+x_2^*)}{h}}  u_h \Big(M_{-}(x) \Big)  u_h \Big(M_{+}(x)\Big) \dd x \Big|\\ &\qquad + \mathscr{O}(e^{-\frac{d(L-R-3\eta)+ d(\ell - L + R + 2 \eta)}{h}}),
\end{align*}
where $M_{\pm}(x) =\sqrt{(x_1 \pm \frac \ell 2 )^2 + (x_2 + x_2^*)^2}$. By Lemma \ref{lem.phase} we have
\[ s_h^{\sigma \sigma'}(x+x^*) = S_h(\ell) + \frac{1}{2} \nabla^2 s_0(x^*) (x,x) + iC \sqrt{h} x_2 + \mathscr{O}(h|x| + |x|^3), \] 
and using the Laplace method,
 \begin{align*}
| \langle \Psi_\alpha^\sigma, \Psi_\beta^{\sigma'} \rangle | =  \mathscr{O} \Big( h^{\frac 12} e^{-\frac{S_h(\ell)}{h}} \Big) + \mathscr{O}\Big(e^{-\frac{d(L-R-3\eta)+ d(\ell - L + R + 2 \eta)}{h}}\Big),
\end{align*}
where we used that $M_{\pm}(x) = R \pm \frac{\ell x_1}{2R} + \frac{x_2^* x_2}{R} + \mathscr{O}(x^2)$ and \eqref{eq.uhR}.
\end{proof}

\begin{lemma}\label{lem.wh1}
Let $\alpha$, $\beta \in V$ be such that $L \leq |\alpha - \beta| = \ell \leq L(1+\varepsilon)$. Then:
\begin{enumerate}
\item When $e_0 \neq \frac 12$ we have
\[ \mathbf W_{\alpha \beta} = (-1)^{m_-} e^{\frac{iB}{2h} \alpha \wedge \beta} 2w_\ell^{--} +  \mathscr{O}(\mathcal E')\]
where
\[ w_\ell^{--} = h^2 \int_{\R} e^{-\frac{iB\ell x_2}{2h}} \Phi^- \Big(\frac \ell 2, x_2 \Big) \partial_1 \Phi^- \Big( \frac{\ell}{2}, x_2\Big)  \dd x_2, \]
and $\mathcal E' = h^{\frac 52} e^{- \frac{S_h(\ell)}{h}} + e^{-\frac{d(L-R-3\eta) + d(\ell - L + R + 2 \eta)}{h}} + e^{-\frac{2d(L-R-3\eta)}{h}}$.
\item When $e_0 = \frac 12$ we have
\[\mathbf W_{\alpha \beta}^{\sigma \sigma'} =(-1)^{m_{\sigma'}} e^{\frac{iB}{2h} \alpha \wedge \beta} e^{i (m_\sigma - m_{\sigma'}) \mathrm{arg}(\beta - \alpha)}  \big( w^{\sigma \sigma'}_\ell + w^{\sigma' \sigma}_\ell \big) + \mathscr{O}(\mathcal E'),\]
where
\[w^{\sigma \sigma'}_\ell =  h^2 \int_{\R} e^{-\frac{iB\ell x_2}{2h}} \Phi^\sigma \Big(\frac \ell 2, x_2 \Big) \partial_1 \Phi^{\sigma'} \Big( \frac{\ell}{2}, x_2\Big)  \dd x_2. \]
\end{enumerate}
\end{lemma}

\begin{proof}
When $e_0 \neq \frac 12$ the coefficient $\mathbf W_{\alpha \beta}$ is equal to $\mathbf W_{\alpha \beta}^{--}$. Thus, we can focus on item $(2)$. We use translations and rotations to center the problem as in the proof of Lemma \ref{lem.scalarproduct}. This amounts to change $\alpha$ to $(- \frac{\ell}{2},0)$ and $\beta$ to $(\frac{\ell}{2},0)$. We obtain, with the notation $\bar{\mu} = \frac{\mu_\sigma + \mu_{\sigma'}}{2}$,
\begin{align*}
\langle (\Lh - \bar \mu) \Psi_\alpha^\sigma, \Psi_\beta^{\sigma'} \rangle &= \langle (\Lh - \bar \mu) \tau_\alpha^B \chi_\eta \Phi^\sigma, \tau_\beta^B \chi_\eta \Phi^{\sigma'} \rangle \\
&=e^{\frac{iB}{2h} \alpha \wedge \beta}\langle (\Lh - \bar \mu) \tau_{\frac{\alpha - \beta}{2}}^B \chi_\eta \Phi^\sigma, \tau_{\frac{\beta- \alpha}{2}}^B \chi_\eta \Phi^{\sigma'} \rangle\\
&=e^{\frac{iB}{2h} \alpha \wedge \beta + i (m_\sigma - m_{\sigma'})arg(\beta - \alpha)} \langle (\Lh - \bar \mu) \tau_{(-\frac  \ell 2,0)}^B \chi_\eta \Phi^\sigma, \tau_{(\frac \ell 2,0)}^B \chi_\eta \Phi^{\sigma'} \rangle.
\end{align*}
We now introduce the notation $\chi_\ell(x) = \chi_\eta(x_1 + \frac \ell 2, x_2)$, $\chi_r(x) = \chi_\eta (x_1 - \frac\ell 2,x_2)$, and $\Phi_\ell^\sigma = \tau_{(-\frac  \ell 2,0)}^B \Phi^\sigma$, $\Phi_r^{\sigma'} = \tau_{(\frac  \ell 2,0)}^B \Phi^{\sigma'}$, so that
\begin{equation}\label{eq.1441}
\langle (\Lh - \bar \mu) \Psi_\alpha^\sigma, \Psi_\beta^{\sigma'} \rangle = e^{\frac{iB}{2h}\alpha \wedge \beta} e^{i(m_\sigma - m_{\sigma'})arg(\beta - \alpha)} \tilde w^{\sigma \sigma'},\end{equation}
 with $ \tilde w^{\sigma \sigma'} = \langle (\Lh - \bar \mu) \chi_\ell \Phi_\ell^{\sigma}, \chi_r \Phi_r^{\sigma'} \rangle$.
We can rewrite $\tilde w ^{\sigma \sigma'}$ using the integration by part argument from \cite{HS}. Since $\Phi_\ell^\sigma$ is an eigenfunction for the single-obstacle problem, we have
\begin{align*}
\tilde w^{\sigma \sigma'} &=\langle \big[ \Lh, \chi_\ell ] \Phi_\ell^\sigma, \chi_r \Phi_r^{\sigma'} \rangle + \big( \mu_\sigma - \bar \mu \big) \langle \chi_\ell \Phi^\sigma_\ell, \chi_r \Phi^{\sigma'}_{r} \rangle \\
& = \langle \big( p_h \cdot \big[ p_h, \chi_\ell \big] + [p_h, \chi_\ell \big] p_h \big) \Phi_\ell^{\sigma} , \chi_r \Phi_r^{\sigma'} \rangle + \frac{\mu_\sigma - \mu_{\sigma'}}{2} \langle \chi_\ell \Phi^\sigma_\ell, \chi_r \Phi^{\sigma'}_{r} \rangle.
\end{align*}
The second term is estimated using $|\mu_\sigma - \mu_{\sigma'}| \leq Ch^2$ and Lemma \ref{lem.scalarproduct}. Hence,
\begin{equation}
\tilde w^{\sigma \sigma'} =  \langle \big( p_h \cdot \big[ p_h, \chi_\ell \big] + [p_h, \chi_\ell \big] p_h \big) \Phi_\ell^{\sigma} , \chi_r \Phi_r^{\sigma'} \rangle + \mathcal E,
\end{equation}
with $\mathcal E = \mathscr{O}( h^{5/2} e^{-\frac{S_h(\ell)}{h}} + e^{-\frac{d(L-R-3\eta) + d(\ell - L + R + 2 \eta)}{h}})$.
A partial integration gives
\begin{align*}
\tilde w^{\sigma \sigma'} &=  \langle \big[p_h, \chi_\ell \big] \Phi_\ell^{\sigma}, p_h \chi_r \Phi_r^{\sigma'} \rangle +  \langle \big[ p_h, \chi_\ell \big] p_h \Phi_\ell^{\sigma}, \chi_r \Phi_r^{\sigma'} \rangle + \mathcal E \\
&= -i h \int_{\R^2} (\nabla \chi_\ell) \Phi_\ell^{\sigma} \cdot \overline{p_h(\chi_r \Phi_r^{\sigma'})} + p_h \Phi_\ell^{\sigma} \cdot \overline{(\nabla \chi_\ell) \chi_r \Phi_r^{\sigma'}} \dd x + \mathcal E.
\end{align*}
We want replace $\chi_r$ by $1$ and add a $\mathbf{1}_{x_1 >0 }$ using the support of $\nabla \chi_\ell$. This is not free, but on the support of $\nabla \chi_{\ell}$, if either $x_1<0$ or $\chi_r \neq 1$, then $\Phi_\ell^{\sigma}$ and $\Phi_r^{\sigma'}$ are both smaller than $e^{-\frac{d(L-R-2\eta)}{h}}$ (Lemma \ref{lem.decay}). Therefore,
\begin{align*}
\tilde w^{\sigma \sigma'} = -i  h \int_{x_1>0}  (\nabla \chi_\ell) \Phi_\ell^{\sigma} \cdot \overline{P \Phi_r^{\sigma'}} + P \Phi_\ell^{\sigma} \cdot \overline{(\nabla \chi_\ell) \Phi_r^{\sigma'}} \dd x + \mathcal E + \mathscr{O}(e^{-\frac{2d(L-R-3\eta)}{h}}).
\end{align*}
We integrate by part again, and obtain a boundary term
\begin{align*}
\tilde w^{\sigma \sigma'} &= i h \int_{x_1>0} \chi_\ell \nabla ( \Phi_\ell^{\sigma} \overline{P \Phi_r^{\sigma'}}) \dd x + i h \int_{x_1=0} \chi_\ell \Phi_\ell^{\sigma} \overline{(-ih \partial_1 + B x_2 /2) \Phi_r^{\sigma'}} \dd x_2 \\ & \quad +i h \int_{x_1>0}  \chi_\ell \nabla ( \overline{\Phi_r^{\sigma'}} P \Phi_\ell ^\sigma) \dd x +i h \int_{x_1=0} \chi_\ell (-ih\partial_1+ B x_2 /2) \Phi_\ell^\sigma \overline{\Phi_r^{\sigma'}} \dd x_2 \\ &\qquad + \mathcal E + \mathscr{O}(e^{-\frac{2d(L-R-3\eta)}{h}}).
\end{align*}
Note that
\begin{align*}
-ih \Big( \nabla ( \Phi_\ell^\sigma \overline{P\Phi_r^{\sigma'}} ) + \nabla ( \overline{\Phi_r^{\sigma'}} P\Phi_\ell^\sigma) \Big) = \overline{\Phi_r^{\sigma'}} \Lh \Phi_\ell^\sigma - \Phi_\ell^\sigma \overline{\Lh \Phi_r^{\sigma'}}  = (\mu_\sigma - \mu_{\sigma'}) \Phi_\ell^\sigma \overline{\Phi_r^{\sigma'}},
\end{align*}
because $\Phi_\ell^\sigma$ and $\Phi_r^{\sigma'}$ are eigenfunctions of $\Lh^0$. We deduce
\begin{align*}
\tilde w^{\sigma \sigma'} &=  h^2 \int_{x_1=0} \chi_\ell \big(\overline{\Phi_r^{\sigma'}} \partial_1 \Phi_\ell ^\sigma -  \Phi_\ell^\sigma \overline{\partial_1 \Phi_r^{\sigma'}} \big) \dd x_2 + (\mu_\sigma - \mu_{\sigma'}) \int_{x_1 >0} \Phi_\ell^{\sigma}\overline{ \Phi_r^{\sigma'}} \dd x \\ & \qquad+ \mathcal E + \mathscr{O}( e^{-\frac{2d(L-R-3\eta)}{h}}).
\end{align*}
The second integral can be estimated as $\mathcal E$ using $|\mu_\sigma - \mu_{\sigma'}| \leq C h^2$ and Lemma \ref{lem.scalarproduct} again. In the first integral, we have $\chi_\ell(0,x_2) = 1$ unless both $\Phi_\ell^{\sigma}(0,x_2)$ and $\Phi_r^{\sigma'}(0,x_2)$ are smaller than $e^{-\frac{d(L-R-2\eta)}{h}}$ and therefore
\begin{align*} \tilde w^{\sigma \sigma'} &= h^2 \int_{\R} e^{-\frac{iB \ell x_2}{2h}} \Big( \overline{\Phi^{\sigma'}}\Big(-\frac \ell 2, x_2 \Big) \partial_1 \Phi^\sigma \Big( \frac{\ell}{2}, x_2 \Big) -\Phi^\sigma \Big(\frac \ell 2, x_2 \Big) \overline{\partial_1 \Phi^{\sigma'}}\Big( - \frac{\ell}{2}, x_2 \Big)  \Big) \dd x_2 \\
& \quad + \mathscr{O}(\mathcal E + e^{-\frac{2d(L-R-3\eta)}{h}}).
\end{align*}
We use that $\overline{\Phi^\sigma}(x_1,x_2) = \Phi^\sigma(x_1,-x_2) = (-1)^{m_\sigma} \Phi^\sigma(-x_1,x_2)$ to obtain
\[ \tilde w^{\sigma \sigma'} = (-1)^{m_{\sigma'}} \big( w^{\sigma \sigma'}_\ell + w^{\sigma' \sigma}_\ell \big) + \mathscr{O}(\mathcal E + e^{-\frac{2d(L-R-3\eta)}{h}}). \]
The result follows using \eqref{eq.1441} and \eqref{eq.Wabs}.
\end{proof}

\begin{lemma}\label{lem.wh2}
For all $\ell \geq L$, let $K_\ell = \frac{\ell}{2R} + \sqrt{\frac{\ell^2}{4 R^2} -1}$. There exists a $C_\ell \in \mathbf C$ independent of $h$ and with non-vanishing real part such that:
\begin{enumerate}
\item When $e_0 \in \big( - \frac 12, \frac 12 \big]$,
\[w_\ell^{--} = C_\ell K_\ell^{-2e_0}  h e^{-\frac{S_h(\ell)}{h}} (1+o(1)),\]
as $h \to 0$ along a $e_0$-sequence.
\item When $e_0 \in \big( - \frac 12, \frac 12 \big]$,
\[ w_\ell^{++} = C_\ell K_\ell^{2e_0} h e^{-\frac{S_h(\ell)}{h}} (1+o(1)),\]
as $h \to 0$ along a $e_0$-sequence.
\item When $e_0 = \frac 12$ and $\sigma \neq \sigma'$, 
\[w_\ell^{\sigma \sigma'} = C_\ell  h e^{-\frac{S_h(\ell)}{h}} (1+o(1)),\]
as $h \to 0$ along a $\frac 12$-sequence.
\end{enumerate}
\end{lemma}

\begin{proof}
We use the explicit formula from Lemma \ref{lem.decay},
\begin{equation}
\Phi^\sigma(x_1,x_2) = \Big( \frac{BR^2}{h} \Big)^{ \frac{\Theta_0}{4}} \exp \Big( - \frac{\varphi_\sigma(x)}{h} \Big) u_h(|x|),
\end{equation}
with $\varphi_\sigma(x) = \frac{B}{4}(x_1^2 + x_2^2 - R^2) - h m_\sigma \ln \big( \frac{x_1+ix_2}{R} \big)$ and
\begin{equation}
\partial_1 \Phi^\sigma  = \Big( \frac{BR^2}{h} \Big)^{ \frac{\Theta_0}{4}}\exp \Big(  - \frac{\varphi_\sigma}{h} \Big) \Big( \frac{x_1}{|x|} \partial_r u_h - h^{-1} \partial_1 \varphi_\sigma  u_h \Big).
\end{equation}
Using these formulas we find
\begin{align*}
 w_\ell^{\sigma \sigma'} =  h^2  \Big( \frac{BR^2}{h} \Big)^{ \frac{\Theta_0}{2}} \int_{\R} e^{- \frac{s_h^{\sigma \sigma'}( 0, x_2)}{h}}  u_h(M(x_2)) \omega(x_2) \dd x_2,
\end{align*}
with $M(x_2) = \sqrt{x_2^2 + \ell^2/4}$ and
\begin{align*}
\omega(x_2) &= \frac{\ell}{\sqrt{\ell^2 + 4 x_2^2}} \partial_r u_h(M(x_2)) + \Big( \frac{m_{\sigma'}}{\ell/2 + i x_2} - \frac{B \ell}{4h} \Big) u_h(M(x_2)) \\
s_h^{\sigma \sigma'}(0,x_2) &= \frac{B \ell^2}{8} - \frac{BR^2}{2} + \frac{B x_2^2}{2} - h(m_\sigma + m_{\sigma'}) \log \Big( \frac{\ell/2 + i x_2}{R}\Big) + \frac{iB \ell x_2}{2}.
\end{align*}
The main properties of the function $s_h^{\sigma \sigma'}$ and its principal part $s_0$ are summarized in Lemma \ref{lem.phase}. In particular it has a unique critical point in the lower half complex plane at $x_2^* = -i \sqrt{\frac{\ell^2}{4} - R^2}$. Since the involved functions are analytic, we can make a complex translation $x_2 \mapsto x_2 + x_2^*$ to obtain
\begin{align*}
w_\ell^{\sigma \sigma'} = h^2  \Big( \frac{BR^2}{h} \Big)^{ \frac{\Theta_0}{2}} \int_{\R} e^{- \frac{s_h^{\sigma \sigma'}( 0, x_2 + x_2^*)}{h}} u(M(x_2 + x_2^*))  \omega(x_2 + x_2^*) \dd x_2.
\end{align*}
We estimate this integral using the Laplace method. We give a few details due to the complexity of the phase involved. Since the phase is minimal at $x_2=0$, and using \eqref{eq.uhr}, we can restrict the integral to $|x_2| \leq C h^{1/4}$. We then rescale with $x_2 = \sqrt{\frac{hN}{B(1-N)}} t$, where we define $N = \frac{4R^2}{\ell^2}$. With this notation we have
\[M(x_2 + x_2^*) = R - i t \sqrt{\frac hB} + \mathscr{O}(h), \]
and therefore
\begin{multline*}
  w_\ell^{\sigma \sigma'} = - \frac{ h^{2}}{2}  \Big( \frac{BR^2}{h} \Big)^{ \frac{\Theta_0+1}{2}} e^{-s_h^{\sigma \sigma'}(0,x_2^*) /h } \\ \int_{\lbrace | t| <C h^{- 1/4} \rbrace} e^{- \frac{1-\sqrt{1-N}}{\sqrt{1-N}}(t^2 +2it \sqrt{\Theta_0})} u_h \Big(R - i t \sqrt{\frac hB} \Big)^2  \big( 1 + o(1) \big)\dd t . 
\end{multline*}
Here we used Lemma \ref{lem.phase} to estimate the phase $s_h^{\sigma \sigma'}$, and Lemma \ref{lem.decay} to show that, in $\omega$, the term in $\partial_r u_h$ is smaller than $h^{-1} u_h$. Now note that in this regime $|t|< C h^{-1/4}$, the estimate \eqref{eq.uhR} remains with $-it$ instead of $t$. Therefore,
\[ w_\ell^{\sigma \sigma'} = - \frac{K_0^2 B h}{2} e^{-s_h^{\sigma \sigma'}(0,x_2^*) /h } \int_{\lbrace | t| <C h^{- 1/4} \rbrace} e^{- \frac{1-\sqrt{1-N}}{\sqrt{1-N}}(t^2 +2it \sqrt{\Theta_0})} v(-it) ^2  \big( 1 + o(1) \big)\dd t . \]
Also note that $v(-it)$ belongs to the Schwarz class, since it is the Fourier transform of a Schwarz function. Thus, the integral is asymptotically equal to a constant, which only depends on $N$. We obtain a constant $\tilde{C}_\ell \neq 0$ such that
\begin{equation}
w_h^{\sigma \sigma'}(\ell) = \tilde C_\ell h e^{-\frac{s_h^{\sigma \sigma'}(0,x_2^*)}{h}} ( 1+ o(1)).
\end{equation}
The critical value of $s_h$ is equal to
\[s_h^{\sigma \sigma'}(0,x_2^*) = S_h(\ell) - h ( 2 \mathcal C_2 + m_\sigma + m_{\sigma'} - 2 \xi_h) \ln (K_\ell), \]
where we recall that $m_- = \xi_h - e_0 + o(1)$ and $m_+=m_-+1$. The result follows with a new $C_\ell$.
\end{proof}

\section{Error terms}\label{sec.errors}

We recall the order of the error terms encountered,
\[ \mathcal E_{\varepsilon,L} = e^{- \frac{2S_h(L)}{h}} + e^{-\frac{2d(L-R-4\eta)}{h}}+ e^{-\frac{S_h(L(1+\varepsilon))}{h}} + e^{-\frac{d(L-R-4\eta) + d(\varepsilon L + R + 2 \eta)}{h}}. \]
We prove here that these errors are indeed small enough.

\begin{lemma}\label{lem.errors}
If $L > \max (6R, \frac{4}{\varepsilon^2} R)$ and if $\eta$ is small enough, then there exists $\kappa >0$ such that
\[ \mathcal{E}_{\varepsilon,L} = \mathscr{O} \Big( e^{-\frac{S_h(L)(1+ \kappa)}{h}} \Big) .\]
\end{lemma}

\begin{proof}
We recall that $S_h(\ell) = S_0(\ell) - \sqrt h S_1(\ell) + \mathscr{O}(h)$ with
{\small{
\begin{equation*}
S_0(\ell) = \frac{B \ell}{4}\sqrt{\ell^2-4R^2} -  BR^2 \ln \Big( \frac{\ell+ \sqrt{\ell^2-4R^2}}{2R}  \Big), \quad S_1(\ell) = 2\sqrt{\Theta_0 BR^2} \ln \Big( \frac{\ell + \sqrt{\ell^2 - 4R^2}}{2R}\Big).
\end{equation*}
}}
 Moreover, the function $S_0$ can be rewritten as
\begin{equation}\label{eq.S0int}
S_0(r) = 2 BR^2 \int_1^{r/2R} \sqrt{u^2-1} \dd u,
\end{equation}
and the Agmon distance as
\begin{equation}\label{eq.dint}
d(r) = \frac{BR^2}{2} \int_1^{r/R} \Big( t - \frac 1t \Big) \dd t.
\end{equation}
In particular, $S_0$ and $d$ are positive increasing functions on $[R,\infty)$. Therefore, it is enough to prove that
\begin{equation}\label{eq.2ds}
2 d(L-R) > S_0(L),
\end{equation}
and
\begin{equation}\label{eq.dds}
d(L-R) + d(\varepsilon L + R) > S_0(L).
\end{equation}
To prove these, we use $\sqrt{u^2-1} < u$ in \eqref{eq.S0int} and $t - \frac 1 t > t-1 $ in \eqref{eq.dint}. This provides us the upper bound
\begin{equation}
S_0(L) \leq \frac{BL^2}{4} - BR^2,
\end{equation}
and the lower bound
\begin{equation}
d(r) \geq \frac{B}{4} (r-R)^2.
\end{equation}
Thus \eqref{eq.2ds} is satisfied when 
\[ \frac B2 (L-2R)^2 > \frac{BL^2}{4} - BR^2, \quad {\rm{i.e.}} \quad \frac{L^2}{4} - 2LR+ 3R^2 >0, \]
for which it is enough that $L > 6R$. The other inequality \eqref{eq.dds} is satisfied when
\[ \frac B4 (L-2R)^2 + \frac{\varepsilon^2 BL^2}{4} > \frac{BL^2}{4} - BR^2, \quad {\rm{i.e.}} \quad LR < \frac{\varepsilon^2 L^2}{4} + 2R^2, \]
for which it is enough that $L > \frac{4}{\varepsilon^2} R $.
\end{proof}

\bibliographystyle{plain}
\bibliography{biblio}

\end{document}